\documentclass[12pt]{article}
\usepackage{graphicx,color}
\usepackage[utf8]{inputenc}
\usepackage{a4,amssymb,amsmath}

\usepackage{newtxtext}

\numberwithin{table}{section}

\newtheorem{thm}{Theorem}[section]  
\newtheorem{prop}[thm]{Proposition} 
\newtheorem{lemma}[thm]{Lemma}       
\newtheorem{coro}[thm]{Corollary}   
\newtheorem{remk}[thm]{Remark}      

\def\eref#1{(\ref{#1})}         
\def\sref#1{Sect.~\ref{#1}}
\def\aref#1{App.~\ref{#1}}
\def\pref#1{Prop.~\ref{#1}}
\def\lref#1{Lemma~\ref{#1}}
\def\rref#1{Remark~\ref{#1}}

\def\cref#1{Cor.~\ref{#1}}

\def\fref#1{Footnote~\ref{#1}}

\def\mystrut#1#2{\vrule height #1pt depth #2pt width 0pt} 

\def\bea#1{\begin{eqnarray}\label{#1}}
\def\eea{\end{eqnarray}}
\def\ba{\begin{array}}
\def\ea{\end{array}}
\def\qed{\hfill $\square$}

\def\gg{{\mathfrak{g}}}
\def\hh{{\mathfrak{h}}}
\def\mm{{\mathfrak{m}}}

\def\eps{\varepsilon}
\def\ka{\kappa}
\def\la{\lambda}

\def\ga{\gamma}

\def\pa{\partial}
\def\sumno{\sum\nolimits}
\def\inv{^{-1}}
\def\ol{\overline}
\def\ul{\underline}
\def\ad{\mathrm{ad}}
\def\wt{\widetilde}
\def\wh{\widehat}
\def\sfrac#1#2{\hbox{\large{$\frac{#1}{#2}$}}}

\def\Vert{\big\vert}

\def\erw#1{\langle #1\rangle}
\def\Erw#1{\big\langle #1\big\rangle}
\def\wick#1{\colon\!\!#1\!\colon}
\def\ww{_{\rm int}}
\def\u{\mathfrak{u}}
\def\su{\mathfrak{su}}
\def\eins{\mathbf{1}}
\def\lra{\leftrightarrow}
\def\ddelta{\boldsymbol{\delta}}
\def\id{\mathrm{id}}
\def\ben{\begin{enumerate}}
\def\een{\end{enumerate}}
\def\bpm{\begin{pmatrix}}
\def\epm{\end{pmatrix}}

\def\ua{{\ul a\hspace{.8pt}}} \def\ub{{\ul b\hspace{.8pt}}}
\def\uc{{\ul c\hspace{.8pt}}}

\numberwithin{equation}{section}

\title{Hidden gauge invariance}

\author{Karl-Henning Rehren
\\[5pt]
{\small
Institut für Theoretische Physik, 
Georg-August-Universität Göttingen, 37077 Göttingen, Germany} \\
{\small krehren@uni-goettingen.de}}

\begin{document}

\maketitle

\begin{abstract} The role of gauge invariance is reconsidered by ``deriving
  it without assuming it'' within an autonomous approach to
  interactions of Standard Model particles. In this approach, the
  renormalizable interactions are purely constrained by quantum
  principles, notably the representation on a Hilbert space, which
  forces interactions to be ``string-localized''. The constraints
  arise by imposing that the perturbation theory of the S-matrix be
  insensitive to the 
  string-localization at tree level. With the string-localized
  interactions thus selected, S-matrices in fact coincide with
  S-matrices of local interactions in the gauge theory approach on
  indefinite state spaces.

To some surprise, bosonic interactions fulfilling the constraints
  at first and second perturbative order enjoy an emergent but possibly
  ``hidden'' gauge invariance that implies fulfilment at tree level at all
    orders.

Hidden gauge invariance relates the autonomous interactions to
    Standard model interactions. It is uncovered via redefinitions of
    quantum fields. It is explained in which sense it is exact and
    unbroken even in the presence of massive vector bosons.

The limitation to tree level and to the adiabatic limit is
  sufficient as a tool to single out the admissible interactions. An
  extension to loop level is outside the scope of the paper; but
  advances in the ongoing program are reported.

The results give new evidence that particle physics with
  massless and massive vector bosons can be implemented without
  indefinite state spaces and ghosts.
\end{abstract}


\bigskip

\hfill {\sl In memory of Ivan T. Todorov (1933-2025)}

\section{Introduction}
\label{s:intro}
It is widely agreed that gauge invariance is extremely successful
as a criterium to select classical Lagrangians for the Standard Model
interactions. But several questions remain unanswered: \emph{Why} 
gauge invariance in the first place? Is there any operational meaning
of gauge invariance when gauge transformations do not affect observables?
Can the notorious trouble with canonical quantization of gauge potentials
(indefinite metric, unphysical degrees of freedom) be avoided?

Ever since Jordan \cite{Jo} and Dirac \cite{Di}, physicists and philosophers
of science (e.g., \cite{OP,FMS2,Ly,Ear1,BFetal}) have worried about
these and related questions. Notably, in \cite{FMS2,Hi,Ki} it
was demonstrated how the ``Higgs mechanism'' can be understood without
need to refer to a gauge symmetry. Concerning the observable content of
classical gauge-invariant Lagrangians, inspiring discussions (and an
extended guide to the literature) can be found in \cite{BFetal,Fra}, see
\rref{rk:DFM} and \sref{s:conc}. 

We want to contribute a new point of view to the discussion. It is
``genuinely quantum'', namely, it is ultimately the need for a Hilbert 
space formulation of the quantum field theory (not addressed in
\cite{Fra}) together with 
  locality that forces renormalizable interactions to be what they
are. Gauge invariance rather \emph{emerges} as a (possibly
hidden) feature of admissible Hilbert space interactions.

The meaning of ``admissible'', and the first sketch of ``hidden gauge
invariance'' in Item~6 below require some explanations of the setup.
It proceeds from an ``autonomous approach'' to particle interactions.

\paragraph{1. The autonomous approach.}
The ``autonomous approach'' to the interactions of particles
(\cite{aut,GGM,chir,weak,infra}) results in a highly constrained method
to select admissible interaction densities, to be used in quantum
perturbation theory. It explains several pertinent physical features
of the Standard Model, where gauge theory falls short, see the quoted
literature and \rref{rk:dress}, \rref{rk:LV-YM}, and \aref{a:chir}. 

``Admissible'' means, first of all, that the interactions densities
are renormalizable and are defined on the physical Hilbert space of the
particles. (The decisive, most restrictive conditions will be stated below.)
For particles of helicity or spin 1, this means that one cannot use
local vector potentials: Covariant quantization of massless vector
potentials $A_\mu(x)$ requires state spaces of indefinite metric,
while the massive Proca vector field $B_\mu(x)$ does not admit
power-counting renormalizable interactions because its short-distance (UV)
fluctuations are too strong (the two-point function has scaling degree 2). 

The autonomous approach is committed to quantum principles (Hilbert
space, covariance and locality). The Hilbert space and covariance
principles provide substitutes for the problematic fields $A_\mu$ and $B_\mu$,
given by ``mildly non-local'' (``string-localized'') potentials, see
Item~2 below. However, the S-matrix must not depend on this
non-locality. This highly nontrivial condition replaces a ``gauge
principle'' to select the admissible interactions.

\paragraph{2. String-localization.}
String-localized potentials for massless vector bosons (``photons'')
and massive vector bosons (MVBs) are defined by the same formula
\bea{AF}
A_\mu(c,x) := I_c^\nu(F_{\mu\nu})(x) \equiv \int d^4y \,
c^\nu(y)F_{\mu\nu}(x+y), 
\eea
where $F_{\mu\nu}$ is the local Maxwell tensor defined on the Fock
space over the unitary Wigner representation of helicity $\pm1$ in the
massless case, and $F_{\mu\nu}:= \pa_\mu B_\nu-\pa_\nu B_\mu$ on the
Wigner Fock space of spin 1 in the massive case, satisfying the free
field equations. For massive vector
bosons, one defines also 
\bea{phic}
\phi(c,x):= I_c^\nu(B_\nu)(x)\equiv \int d^4y\, c^\nu(y) B_\nu(x+y).
\eea
The function $c^\mu(y)$ is required to satisfy
\bea{pac}
\pa_\mu c^\mu(y)=\delta(y),
\eea
which entails the identity for the integral transform
$I^\nu_c(f)(x) = \int dy\, c^\nu(y)f(x+y)$
\bea{Ipa}
I_c^\nu(\pa_\nu f) = \pa_\nu I_c^\nu(f) =-f.
\eea
The identity \eref{Ipa} in turn, together with the field equations
$\pa_\ka F_{\mu\nu}+\hbox{cycl.}=0$ for the local field tensor in both
cases (massless and massive) implies 
\bea{AcF}
\pa_\mu A_\nu(c,x)-\pa_\nu A_\mu(c,x)= F_{\mu\nu}(x).
\eea
In the massive case, it also follows
\bea{paphi}
A_\mu(c)=B_\mu+\pa_\mu\phi(c).
\eea

The potentials $A_\mu(c)$ have UV dimension 1 and serve as Hilbert
space counterparts of gauge potentials, admitting power-counting
renormalizable interactions
\footnote{\label{}That power-counting is applicable as a criterium for
  renormalizability of string-localized interactions, is a consequence
  of work by Gass \cite{G}.} 
among each other and with scalar fields and spinor currents. We denote
by $L\ww(c,x)$ the interaction densities.  

There are many functions $c^\mu$ satisfying \eref{pac}, but their supports
necessarily extend to infinity. Yet, they may (but need not) be narrow cones,
which explains the terminology ``string-localized''.

Under infinitesimal changes $\delta c^\mu$ of the function $c^\mu$,
one has $\delta_c(F)=\delta_c(B)=0$ and 
\bea{deltac}
\delta_c(A_\mu(c,x)) = \pa_\mu w(\delta c,x), \quad
\delta_c(\phi(c,x))= w(\delta c,x),
\eea
where $\phi(c)$ exists only if $m>0$. The explicit form $w\equiv w(\delta c)=
-I^\nu_c(\delta_c A_\nu(c))$ of the
``infinitesimal variation field'' \cite[Lemma~2.1]{weak} will never be 
needed.

\paragraph{3. String-localization and string-independence.}
String-localization is a problem because locality of the perturbation
theory is at stake with a non-local interaction density. This is the
origin of the most important constraints on admissible interactions:
the resulting S-matrix 
\bea{SL}
S_{L\ww(c)}= Te^{i\int dx\, L\ww(c,x)}
\eea
(also with source terms for observable fields) must be independent of the
string function $c^\mu$. We write this postulate as
\bea{SI}
\delta_c(S_{L\ww(c)}) \stackrel!= 0. \qquad \hbox{\sl
  (String-independence)}
\eea

A more ambitious demand is that there exists another
\emph{manifestly string-independent and local} interaction density
$K\ww$ such that
\bea{SLSK}
S_{L\ww(c)} \stackrel!= S_{K\ww}. \qquad\hbox{\sl (Equality of
  S-matrices)}
\eea
The postulate \eref{SI} is truly autonomous: The fields that may
appear in $L\ww(c)$ are dictated by quantum principles, and the
postulate selects those $L\ww(c)$ as admissible, for which the
string-localization has no effect. In particular, there is no
reference to a ``gauge principle''.  

For \eref{SLSK} in contrast, one has in mind a local density $K\ww$ to
compare with, which may be provided by the standard local (gauge theory)
approaches to particle interactions. It is possibly not defined on a
Hilbert space (massless case), or non-renormalizable (massive case) or both. 

We emphasize that S-matrices \eref{SL} involve only the interaction
density $L\ww$ (a Wick polynomial in free quantum fields). A ``free
Lagrangian'' $L_0$ is not needed, and not present in the autonomous 
approach, because the free fields are not quantized by ``canonical
quantization'', but directly constructed from the Wigner
representations \cite{Wein}. (There will appear some $L_0$ in the
sequel of the paper (notably \cref{c:ILQ} and \sref{s:HGI}), but it is
rather just a tool to simplify the conditions for \eref{SI} and
\eref{SLSK}, reducing a tedious explicit computation to a simple invariance
argument. On the other hand, $L_0$ is indispensible for ``hidden
  gauge invariance'' to be uncovered, see \sref{s:HGI}.) 

\paragraph{4. The limitations to tree level and the adiabatic limit.}
We analyze \eref{SI} and \eref{S=S} only at
\emph{tree level}, since we are interested in \emph{necessary}
conditions to fix the interactions. The validity at loop level presently remains
open. Let us briefly describe the state of affairs. The same
power-counting criterium limiting the freedom of renormalization as
for local interactions was established in \cite{G}. Since
string-localized fields $A(c)$ (both massless and massive) have UV
scaling dimension 1, minimal interactions like $A(c)_\mu j^\mu$ are
defined on a Hilbert space and are renormalizable in both cases. But
this result does not exclude that renormalization ``constants'' may
become string-dependent coefficients.

An ongoing study \cite{MR} on the (kinematically
simplest) model of massive QED scrutinizes the Epstein-Glaser
construction \cite{EG} of the S-matrix (understood as a power series of
operator-valued distributions) and of interacting fields at loop
level. String-localized fields are treated
as distributions in the position $x$ and the string-direction $e$, and
the adiabatic limit (removal of a space-time cutoff) is \emph{not}
assumed. There is considerable 
progress, but some gaps in the argument remain to be filled.

A main problem was the Epstein-Glaser causal factorization to
recursively determine all time-ordered vacuum expectation values
outside the total diagonal $x_1=\dots=x_n$, which seems impossible 
with strings at first sight. An analysis of the analytic properties of 
string-localized fields reveals, however, that causal factorization is needed 
only for the arbitrarily short ``heads'' of strings (segments $x+se$ with 
$s<r$ arbitrarily small), while the extension of the ``tails'' is unique. 
Thus, the recursive construction essentially works as in the local case, 
and extensions to the total diagonal (``(re-)normalization'') are possible
up to $\delta$-functions in $x$-space with string-dependent coefficients.

Their string-dependence is fixed by three normalization conditions
  (``NSI conditions'') on time-ordered correlation functions. Together
  with standard conditions as in the local case (symmetries and Ward
  identities), they guarantee string-independence of the S-matrix. One
  of the conditions is trivially satisfied by the choice of kinematic
  propagators. The other two have been verified for any number of
  loops appearing up to third order in perturbation theory, with work
  on arbitrary order in progress.

Only finitely many string-independent renormalization constants remain.
 Moreover, the power-counting non-renormalizable S-matrix of the local
interaction can be renormalized with the help of the
string-localized S-matrix: the infinite arbitrariness of the
former is reduced by the matching with the latter.

The resulting construction preserves -- at loop level and without
adiabatic limit -- the structural features (the derivation property
\eref{deriv} of obstruction maps, and the ME equation \eref{ME}
underlying hidden gauge invariance), that are instrumental to establish
string-independence of the S-matrix and the identity \eref{magic} for
the interacting fields, see \rref{rk:dress}.

The mentioned NSI conditions are taylored for the massive QED model.
  Their adaption to nonabelian theories may be challenging. The preservation 
at loop level of the structures outlined in the present paper may be
pivotal to identify the appropriate NSI conditions.

We also emphasize that in the present analysis, the ``adiabatic
  limit'' is taken throughout, i.e., we work without a space-time
  cutoff function $g(x)$ of the coupling constant. This is legitimate
  again because we are interested in \emph{necessary} conditions to
  fix the interactions. More precisely, we shall frequently replace
  distributions of the form $T[\pa^y_\mu Y(y),X(x)]$ (to be
  smeared with $g(y)$) by the difference \eref{Om}, that produces
  $\delta$ functions (``obstructions'', see Item~5 and \sref{s:obst}).

In the actual Epstein-Glaser construction of a model (as e.g.,
  in \cite{MR}), cutoff functions are fully taken into account by
  regarding each order of the S-matrix as a distribution in $n$
  variables \cite{MR}. This requires to add interaction terms, that are
  smeared with \emph{derivatives} of the cutoff function and systematically
  produce the subtracted term in \eref{Om}. See, e.g., \cite{AHM,LV}.

It is clear that issues like anomalies (see, however, the
  comment at the end of \aref{a:chir}) and non-perturbative
  effects like confinement are outside the scope of the present
  work. Also, the problem of the adiabatic limit in models with
  massless particles will not be disappear in the autonomous approach.

\paragraph{5. Obstructions.}
The constraints on $L\ww(c)$ (and $K\ww$)
arising from \eref{SI} and \eref{SLSK} can be analysized recursively 
order by order in perturbation theory. The basic mechanism is that
``total derivatives do not contribute'' to the integrals in the perturbative
expansion of the S-matrix. The ``initial'' first-order condition on
$L\ww(c)=gL_1(c) + \frac12g^2L_2(c)+\dots$
($g$ is a coupling constant) is therefore that the
$c$-dependence of $L_1(c)$ is a total derivative, Eq.~\eref{LQ1},
respectively $L_1(c)-K_1$ is a total derivative, Eq.~\eref{LV1}. (For
these features, the derivative in $\delta_c(A)=\pa w$ is
instrumental.
\footnote{\label{}Rivat \cite{why} explains the QED interaction along
  similar lines: Because a covariant Lorentz transformation law for a gauge
  potential in the Wigner Hilbert space of helicity $\pm1$ fails by a derivative
  term, $A_\mu$ must be coupled to a conserved current, so that
  $A_\mu j^\mu$ is Lorentz invariant up to a total derivative.}
)
The problem is that time-ordering does not commute with derivatives,
which results in ``obstructions'' at higher orders, see \sref{s:obst},
which must be cancelled by adding higher-order interactions
$L_2(c),\dots$. The latter are determined (``induced'') by
the obstructions that they must cancel. In principle, this might not
be possible; but if they exist, the obstructions are said to be
``resolved''.

Several admissible cubic interaction densities $L_1(c)$ can be
identified, describing self-interactions and interactions with scalar
fields and spinor currents. Resolvability at second order typically
constrains parameters present in $L_1(c)$, and then $L_2(c)$ exists.
It also happens that for resolvability, one has to add another
piece (involving more particles) to $L_1(c)$, and then proceed to $L_2(c)$.
E.g., self-interactions of massive vector bosons necessarily require a
coupling to a scalar field.

Higher-order interactions $L_n(c)$ ($n\geq 3$) which would violate
power-counting renormalizability should be absent. But the recursive
method becomes impractical to establish \eref{SI} beyond second
order. We shall formulate in
\sref{s:obst} new necessary and sufficient conditions for \eref{SI} and for
\eref{SLSK} (\pref{p:ILQ} and \pref{p:ME}) \emph{at all perturbative
  orders} at tree level.  

Remarkably, when the particle content is specified (including the masses),
all the interactions $L\ww(c)$ -- distinguished by string-independence
{at second order} -- exhibit a strong similarity with the known 
interactions of the Standard Model (electroweak, QCD) \cite{aut}.
We shall give a precise meaning to ``strong similarity'' in \sref{s:HGI}:
namely, there appears a ``hidden'' relation between $L\ww(c)$ and a
gauge-invariant classical Lagrangian, as illustrated in Item~6 below.

Together with systematic properties of the obstructions,
this relation is responsible that the conditions of
\pref{p:ILQ} and \pref{p:ME} are satisfied \emph{at all perturbative
orders} (see Item~7 below and \sref{s:HGI}, and \sref{s:SM} that
this structure prevails in all SM interactions).

\paragraph{6. Hidden gauge invariance -- an example.}

To give an idea of ``hidden gauge invariance'', we present here a simple
example. For details, see \sref{s:SM}, Item~3. \emph{Mutatis
  mutandis}, the same structure is realized in several physically
important models, including all bosonic SM interactions, see \sref{s:SM}.   

The unique autonomous string-localized interaction between one massive
vector boson and one massive scalar particle \cite{Schroer,AHM}
is the string-localized version of the Abelian Higgs model
\cite{Hi}: 
\bea{}
\notag L\ww(c) &=& mg A_\mu(c)(B^\mu H +\phi(c)\pa^\mu H) +
\sfrac{g^2}2A_\mu(c)A^\mu(c) (H^2+m^2\phi(c)^2)\\  &&
- \sfrac12 m_H^2 \big(\sfrac gm
H(H^2+m^2\phi(c)^2) + \sfrac{g^2}{4m^2}(H^2+m^2\phi(c)^2)^2\big).
\eea
Here and everywhere below, Wick ordering is understood but will never
be written. 

``Hidden gauge invariance'' is the identity, valid in the Wick algebra,
\bea{hidden}
L_0 + L\ww(c)= L[A(c),\Phi(c)] \qquad \hbox{with}\quad
\Phi(c):=v+H+im\phi(c) \quad (v=\sfrac mg),
\eea
where $L_0=-\frac14 F_{\mu\nu}F^{\mu\nu} + \frac12m^2 B_\mu B^\mu +
\frac12\pa_\mu H\pa^\mu H -\frac12m_H^2 H^2$ is the free Lagrangian of
the massive Proca field $B$ and the Higgs scalar $H$ (lifted to a Wick
polynomial), and 
\bea{LAPhi-Intro}
L[A,\Phi] =-\sfrac14 F_{\mu\nu}F^{\mu\nu} +\sfrac12
(D_\mu\Phi)^*D^\mu\Phi -\sfrac{m_H^2}{8v^2}(\Phi^*\Phi-v^2)^2
\eea
is the (classical) gauge-theoretic Lagrangian of a massless gauge
potential $A$ coupled to a complex scalar field $\Phi$ with a double-well
potential.

$L[A,\Phi]$ is manifestly gauge-invariant as a functional of $A$ and $\Phi$.
The same holds true when it is lifted to the Wick algebra and
  Wick-algebra-valued gauge coefficients are admitted. Thus, \eref{hidden}
  reveals the hidden gauge invariance of the string-localized interaction
  $L\ww(c)$. Notice that concrete quantum fields are inserted for the
  dummy arguments of $L[A,\Phi]$. The Lagrangian
  $L[A,\Phi]$ has no dynamical meaning, like generating classical
  equations of motion or imposing canonical quantization
  rules. Moreover, the same identity \eref{hidden} holds
when $A(c)$ and $\Phi(c)$ are replaced by arbitrary gauge
transforms $A(c)+ \pa \alpha$ and $e^{ig\alpha}\Phi(c)$. In particular,
the vacuum expectation value of $\Phi(c)$ may be $x$-dependent
and has no physical meaning, see \rref{rk:no-SSB}. In this sense,
the hidden gauge invariance of $L\ww(c)$ is exact and unbroken.
It is just the gauge invariance of the classical Lagrangian
$L[A,\Phi]$ in the rewriting \eref{hidden}, lifted to the Wick algebra,
which holds as well with Wick-algebra-valued gauge parameters.

\paragraph{7. The new role of gauge invariance.}
Admissible interactions involving vector bosons are necessarily
  string-localized with
  arbitrary string functions $c$. They are constrained by (tree-level)
  string-independence \eref{SI} at each perturbative order. The
  solutions to these constraints at first and second order
  (\sref{s:obst}) are essentially fixed by an assumed particle
  content. They turn out to be related to gauge-invariant
  classical Lagrangians of subtheories of the Standard Model
  (\sref{s:SM}; the gauge group of YM is not fixed, but for possible
  anomalies that must be cancelled at loop level, \aref{a:chir}), with
  suitable string-localized quantum fields inserted for the classical
  gauge and matter fields in the same way as \eref{hidden}.

Hidden gauge invariance is not just a surprising ``nice-to-have'' feature 
of admissible interactions $L\ww(c)$. Instead, it is instrumental in proving
string-independence \eref{SI} \emph{at all
  perturbative orders}. The point is that \eref{SI} can be 
formulated as the invariance of $L_0+L\ww(c)$ under a derivation
$\delta_c+\omega_Q$, where $\delta_c$ is the string variation and
$\omega_Q$ is a map induced by the obstructions of the model
(\pref{p:ILQ}) -- and that this derivation turns out to act like an
infinitesimal gauge transformation with string-dependent
field-valued gauge parameters
\footnote{\label{fn:strict}
Strictly speaking, such transformations should not be called ``gauge
transformations''. See \cite[Sect.~2]{Fra}.}
on the string-localized fields $A(c)$ and $\Phi(c)$ in (analogues of)
\eref{hidden}. In this way, hidden gauge invariance secures \eref{SI}
at tree  level at all orders, \pref{p:HGI-LQ}.

Similarly, \pref{p:ME} reformulates the equality of S-matrices \eref{SLSK}
in terms of a field homomorphism $e^{\omega_U}$, induced by the
obstructions of the model (\pref{p:ME}). This transformation takes
string-localized fields to gauge transforms (again with string-dependent
operator-valued parameters) of local fields in a gauge-invariant
$K\ww$. In this way, hidden gauge invariance also secures \eref{SLSK},
\pref{p:HGI-LV}.

Gauge invariance, that is never \emph{imposed} in the autonomous
approach, is therefore not a ``first principle'', but an emerging
feature, that serves as a sufficient condition to make all
obstructions resolvable, precisely in the physically interesting cases. 

Hidden gauge invariance and the nontrivial properties of obstruction
maps ``acting like gauge transformations'' as just outlined,
will be shown to prevail in all (bosonic) autonomous SM
interactions. These consequently provide the starting point for
  an equivalent description of the SM, in which no massless vector
potentials need to be quantized, or ``turned massive'' by a ``Higgs
mechanism''. All vector bosons are quantized in terms of the
string-localized potentials \eref{AF} associated with the free quantum
Maxwell and Proca fields, allowing renormalizable interactions with
MVBs massive from the outset. Yet, a Higgs boson is necessary without
which the obstructions of self-interactions of MVBs cannot be resolved
\cite{chir,weak}.

\paragraph{8. Plan of the paper.}
In \sref{s:obst}, we briefly sketch the ``new organization'' of
obstructions and their cancellation, that leads to the present
results, and point out where it departs from earlier treatments.
We then derive the new formulation (\pref{p:ILQ}) of 
the condition for string-independence \eref{SI} in terms of an ``obstruction
map'' $\omega_Q$. We also present the new formulation (\pref{p:ME})
of the condition for equality of S-matrices \eref{SLSK} in terms of another
obstruction map $\omega_U$, but we refer to \cite{sQCD} for the proof. 

In \sref{s:HGI}, we show how hidden gauge invariance is instrumental
to secure \eref{SI} and \eref{SLSK}.

In \sref{s:SM}, we address the hidden gauge invariance and ensuing
string-independence of all autonomous interactions of the Standard
Model  (with a special treatment of  fermionic couplings).

\section{Obstruction maps and the resolution of obstructions}
\label{s:obst}
\paragraph{1. Obstructions.}
Perturbative QFT proceeds in terms of time-ordered products of
interaction densities, which by Wick's theorem can be expressed in
terms of time-ordered correlation functions (propagators) of free
fields. Propagators for derivatives of fields can be defined ``kinematically'':
$\erw{0\vert T[\pa\varphi(x)\chi(y)]\vert 0}:= \pa^x\erw{0\vert
  T[\varphi(x)\chi(y)]\vert 0}$
-- except when equations of motion allow to express $\pa\varphi$ in
terms of other fields, whose propagators are defined independently. In
this case, time-ordering does not commute with derivatives. For Wick
polynomials $Y$ and $X$, we define the ``obstruction''
\footnote{\label{}When $Y=j^\mu$ is a conserved current, the
  obstruction means the violation of a Ward identity.} 
\bea{Om}
O_\mu(Y(y),X(x)):= \big(T[\pa_\mu Y(y)X(x)] - \pa^y_\mu T[ Y(y)X(x)]\big)
\big\vert^{\rm tree}.
\eea
For basic properties of obstructions, see \aref{a:obst}. At tree
level, they are completely determined by the ``two-point obstructions''
$O_\mu(\varphi(y),\chi(x))$ of linear fields, see \eref{factorwise}, and extend 
like derivations to Wick products in each argument. 
The two-point obstructions of local fields, can only be derivatives of
$\delta(y-x)$, while for string-localized fields, string-integrals
over derivatives of $\delta(x-y)$ may appear. 

We therefore assume that $O_\mu(Y(y),X(x))$ are known for each model,
except that there may be some free parameters in possibly
non-kinematic propagators of some linear fields of higher UV dimension
(such as the Proca field or the derivative of a scalar field).

Let us now study the appearance of obstructions in perturbation
theory. Let $\ddelta$ some infinitesimal variation of the linear fields,
extended as a derivation to the first-order interaction $L_1$. Assume
there is a vector-valued field $Q_1^\mu$ such that $L_1$ and $Q_1$
satisfy the ``initial condition''
\bea{LQ1}
\ddelta L_1(x)=\pa_\mu Q_1^\mu(x).
\eea
Then the first order of the S-matrix $S_{L\ww}$ is automatically
invariant:
\footnote{\label{}Throughout, we assume sufficiently rapid decay of the
  integrands such that integrals over derivatives vanish. This
  assumption is true for all $L(c)$ and $Q$ of  interest  by
  inspection of the free two-point functions and propagators involved.}
$$
ig\, \ddelta \int dx\, L_1(x)= \int dx \, \pa_\mu Q^\mu_1(x) =0.
$$
However, the second order is \emph{not} invariant:
\bea{O2}
\sfrac{-g^2}2\ddelta \iint dy \, dx\, T[L_1(y)L_1(x)] &=&
\sfrac{-g^2}2\iint dy\, dx\, T[\pa^y_\mu
Q^\mu_1(y)L_1(x)] + (x\lra y) \notag \\ &=& \sfrac{-g^2}2
\int dx\int dy\, O_\mu(Q^\mu_1(y),L_1(x)) + (x\lra y)\qquad
\eea
after subtraction of an integral over a derivative. In previous
work, we demanded that this obstruction can be cancelled
(``resolved'') by a higher-order interaction $\frac12g^2 L_2$, such that
the total second-order contribution to the integrand of $\ddelta
S_{L\ww}$ is another derivative of the form
\bea{resold}
i O_\mu(Q^\mu_1(y),L_1(x)) + (x\lra y)  + \ddelta L_2(x)\cdot
\delta(x-y)\stackrel! =  \sfrac12\pa_\mu^xQ_2^\mu(x,y)+(x\lra y).
\eea 
The obstruction would determine $L_2$ and $Q_2$ -- without a guarantee
that they exist.  

\paragraph{2. The new organization of obstructions.}
The pattern for the cancellation of obstructions is not unique.
\eref{O2} is as well cancelled if 
\bea{resnew}
2i \int dy\, O_\mu(Q^\mu_1(y),L_1(x)) + \ddelta L_2(x) \stackrel! =
\pa_\mu Q^\mu_2(x).
\eea
This has the advantage that terms in $O_\mu(Q^\mu_1(y),L_1(x))$ that
are $y$-derivatives need not be considered. In fact, such terms
\emph{do} appear in some models. With the symmetrized pattern \eref{resold},
they contribute to $Q_2$, where they may jeopardize the resolvability
of the obstruction at the next order. With the integrated pattern \eref{resnew},
they are systematically discarded. Moreover, $Q_n$ defined recursively as
in \eref{resnew} will have a single argument $x$, unlike in \eref{resold}.

Therefore, we adopt the new pattern \eref{resnew} and extend it to higher
orders. We define the ``integrated obstruction map'' for vector-valued
Wick polynomials $Y^\mu$
\bea{omega}
\omega_Y (X)(x) := i \int dy \, O_\mu(Y^\mu(y),X(x))
\eea
on Wick polynomials $X$, and rewrite \eref{resnew} as
\bea{ILQ2}
2\omega_{Q_1}(L_1)(x) + \ddelta L_2(x) \stackrel!= \pa_\mu Q^\mu_2(x).
\eea

The following proposition will pave the way to string-independence
(when $\ddelta=\delta_c$ is the string variation), but is formulated
in a context-independent way with an unspecified
  infinitesimal variation $\ddelta$ of the fields.  
  \begin{prop}\label{p:ILQ} (String-independence) \cite[Cor.~2.5]{sQCD}.
    If there exists a vector-valued Wick polynomial $Q^\mu$ of
    sufficiently rapid decay, satisfying  
  \bea{ILQ}
  \ddelta L\ww \stackrel!= \pa_\mu Q^\mu - \omega_Q(L\ww),
  \eea
then the S-matrix
  $S_{L\ww}= T [e^{i\int dx\, L\ww(x)}]$ is invariant (at tree level):
  \bea{SIgen}
  \ddelta S_{L\ww}=0.
  \eea
  The converse is true in the sense of power series: If
$L\ww=gL_1+\frac12g^2L_2+\dots$ and $\ddelta S_{L\ww}=0$, then
$Q=gQ_1+\frac12g^2Q_2+\dots$ exists satisfying \eref{ILQ}.

\end{prop}
\emph{Proof:} Because $\ddelta$ is a derivation, we have 
\bea{dS}
\ddelta S_{L\ww} = \ddelta T [e^{i\int dx\, L\ww(x)}] &=& i \int dy\, T [\ddelta
L\ww(y) e^{i\int dx\, L\ww(x)}]. 
\eea
Therefore, \eref{ILQ} inserted into \eref{dS} implies
$\ddelta S_{L\ww} =0$ by \lref{l:MWIexp}.

Conversely, if $\ddelta
S_{L\ww}=0$, then by \eref{dS} and again \lref{l:MWIexp}, it holds for arbitrary $Q$
\bea{XX}
i \int dy\, T [(\ddelta L\ww+\omega_Q(L\ww) -\pa_\mu Q^\mu)(y)
e^{i\int dx\, L\ww(x)}] =0.
\eea
Now, expand \eref{XX} in $g$. If $Q$ satisfies \eref{ILQ} until order
$g^{n-1}$, the exponential does not contribute to order $g^n$ of
\eref{XX}, which then asserts that $\ddelta L_n + \omega_Q(L\ww)^{(n)}$
is a derivative, where the $n$-th order contribution
$\omega_Q(L\ww)^{(n)}$ involves only $Q_\nu$ ($\nu<n$). This
re-defines $Q_n$, so as to satisfy \eref{ILQ} at order
$g^n$. Thus, \eref{SIgen} implies \eref{ILQ} with $Q$ a power 
series in $g$. \qed 

\newpage

The expansion of \eref{ILQ} gives back the first- and second-order conditions
\eref{LQ1} and \eref{ILQ2}, and at third order
\bea{ILQ3}
3 \omega_{Q_2}(L_1) +
3 \omega_{Q_1}(L_2)  + \ddelta L_3  \stackrel != \pa_\mu Q^\mu_3,
\eea
etc. These conditions recursively determine $L_n$ and $Q_n$ --
provided they exist. If they exist for all $n$, then $\ddelta S_{L\ww}=0$. The
point of \pref{p:ILQ} is, however, that one can find $Q$ that solves
\eref{ILQ} at all orders ``in one stroke'', see \rref{rk:Qwla}. In particular, proving
\eref{ILQ} (or \eref{ILQ0} below) with $L_n=0$ ($n\geq3$) includes proving
\eref{ILQ3} and its higher-order versions with $L_n=0$. 

The analogous proof of \pref{p:ME} (along similar lines) was given in \cite{sQCD}:
\begin{prop}\label{p:ME} (Equality of S-matrices) \cite[Cor.~2.5]{sQCD}. 
  If for two interactions $L\ww(x)$ and $K\ww(x)$, there exists a
  ``mediator field'' $U^\mu(x)$ satisfying the ``mediating equation'' 
  \bea{ME}
 e^{\omega_U}(L\ww)(x)-K\ww(x) = F(\omega_U)(\pa_\mu U^\mu)(x),
  \eea
  where $F(\omega) = \frac{e^\omega-1}\omega = 1+
  \frac12\omega+\frac16\omega^2+\dots$ as a power series of maps,
  then $L\ww$ and $K\ww$ define the same S-matrix (at tree
    level): 
\bea{S=S}
Te^{i\int dx\,L\ww(x)}=Te^{i\int dx\,K\ww(x)}.
\eea
\end{prop}
In the perturbative expansion of \eref{ME}, the ``initial''
first-order condition is
\bea{LV1}
L_1-K_1=\pa_\mu U^\mu_1.
\eea

\begin{remk}\label{rk:dress} 
Interacting fields are (schematically) defined by ``Bogoliubov's formula''
$$\varphi\vert_{L\ww}(x):= T[e^{i\int L\ww}]^*T[\varphi(x)
    e^{i\int L\ww}], $$
giving rise to a Dyson series of retarded integrals.
In string-localized QFT, because the interaction is not strictly
local, locality of the perturbative interacting fields is at stake.
Fortunately, a variant of \pref{p:ME} with source term insertions
in \eref{S=S} allows to address the issue \cite[Cor.~2.6]{sQCD}.
Namely, one can construct so-called ``dressed fields''
$e^{\omega_U}(\varphi)$ such that at tree level
\bea{magic}
\varphi\vert_{L\ww(c)} = (e^{\omega_U}(\varphi))\vert_{K\ww}.
\eea
Unlike interacting fields, dressed fields contain integrals over
  $\delta$-functions generated by obstruction maps, rather than
  over retarded propagators.
Thus, their localization relative to the free Wigner fields can be
manifestly read off, and determines the \emph{relative} localization
of the  resulting interacting fields among each other (because the
local interaction $K\ww$ in  \eref{magic} preserves relative
 localization). Local dressed fields correspond 
 to interacting observables, while string-localized dressed fields
 correspond to string-localized interacting fields.

 Especially dressed charged  fields turn out to be string-localized.
 E.g., the dressed Dirac field of QED is a free Dirac field
 with ``a photon cloud attached'' \cite{infra}. (Dirac \cite{Di} has
 proposed a similar \emph{classical} object in order to quantize
 gauge-invariant quantities, but here it is derived 
 ``from the quantum side''.) This is a most welcome and physically
 necessary feature, because in QED, the Dirac field \emph{must not}
 commute with the electric flux at spacelike infinity in order to comply with
 Gauss' Law.

The dressings in \eref{magic} (computed in terms of the mediator
  field between an autonomous and a local gauge interaction) look very
  similar, but are conceptually very different from dressings as a
  gauge-fixing prescription \cite{LM,HILLM} (see \rref{rk:LV-YM}) or
  as a tool to reformulate classical Lagrangians in terms of
  gauge-invariant fields \cite{BFetal,Fra} (see \rref{rk:DFM}).
\end{remk}

We present two simple corollaries to \pref{p:ILQ} and \pref{p:ME}, that
are interesting because their assumptions turn out to be satisfied
in many models that we are going to study, including all bosonic
interactions of the Standard Model. It is part of the scheme
how hidden gauge invariance secures string-independence and
coinciding S-matrices. 
\begin{coro}\label{c:ILQ} (Tree level understood) \\
  (i) Assume in \pref{p:ILQ} that there is a quadratic field $L_0$
  such that $\pa_\mu Q^\mu = -(\ddelta+\omega_Q)(L_0)$. Then the
  condition \eref{ILQ} is equivalent to
  \bea{ILQ0}
  (\ddelta+\omega_Q)(L_0+L\ww)\stackrel! =0.
  \eea
   (ii) Assume in \pref{p:ME} that there is a quadratic field $L_0$
  such that $\pa_\mu U^\mu = -\omega_U(L_0)$. Then the condition
  \eref{ME} is equivalent to 
  \bea{ME0}
  e^{\omega_U}(L_0+L\ww)\stackrel! =L_0+K\ww.
  \eea
\end{coro}
\emph{Proof:} The proof of (i) trivial. For (ii), use that
$F(\omega_U)\circ\omega_U=e^{\omega_U}-\id$ in \pref{p:ILQ}. \qed

\section{Hidden gauge invariance at work}
\label{s:HGI}

\paragraph{1. String-independence.}
The admissible autonomous interactions $L\ww(c)$ involving massless and
massive vector bosons and their obstruction maps $\omega_Q$ all enjoy
some remarkable properties that we call ``hidden gauge invariance''.
This is unexpected because $L\ww(c)$ were determined by imposing
string-independence \eref{SI} in lowest order of perturbation theory,
\emph{without} assuming any gauge invariance. See the comments in the
Conclusions, \sref{s:conc}. 

Hidden gauge invariance -- if it prevails -- is most efficient to
prove tree-level string-independence at all orders. While the first
order \eref{LQ1} is an initial condition and the second order
obstruction \eref{ILQ2} has to be resolved by hand in order to
determine $L\ww(c)=gL_1(c)+\frac12g^2L_2(c)$, all higher orders
are taken care of by hidden gauge invariance (including the proof that
all higher $L_n(c)=0$).

We present the argument in the setup of \pref{p:ILQ} independent of
the nature of the derivation $\ddelta$ (the string variation $\delta_c$ in the
case of string-localized QFT (sQFT)). The properties of $L\ww$ and $Q$
consist of several parts (always understood at tree level): 
\begin{itemize} \itemsep-0.5mm
\item[(P1)] It holds $\pa_\mu Q^\mu=-(\ddelta+\omega_Q)(L_0)$ for some
  quadratic Wick polynomial $L_0$.
\item[(P2)]
  There are a \emph{classical} Lagrangian $L[A,\Phi]$, and \emph{quantum}
  fields $\wh A$ and $\wh \Phi$ (functions of the fields in $L\ww$) such that
  it holds
\bea{hGI-LQ}
(L_0+L\ww)(x) = L[\wh A, \wh \Phi](x).
\eea
\item[(P3)]
$L[A,\Phi]$ is gauge-invariant under infinitesimal gauge transformations
$\delta_\lambda$ with arbitrary gauge parameters $\lambda$.
\item[(P4)] There exist Wick-algebra-valued, field-dependent gauge
  parameters $\lambda$ such that 
\bea{P4}
  (\ddelta+\omega_Q)(\wh A) = \delta_{\lambda}(\wh A)
  ,\quad
  (\ddelta+\omega_Q)(\wh \Phi) = \delta_{\lambda}(\wh \Phi).
\eea
\end{itemize}
\begin{remk}\label{rk:HGI}
(i) (P2) and (P3) explain the name ``hidden gauge invariance'': $L_0+L\ww$
is gauge invariant only through the appropriate identification of fields
$\wh A$ and $\wh \Phi$, see \eref{hidden} and the examples in \sref{s:SM}.
And only through (P1) and (P4), demanding properties of $\ddelta$ and
the obstruction map $\omega_Q$, hidden gauge invariance turns into a
powerful feature to secure $\ddelta$-invariance of the S-matrix,
cf.~\pref{p:HGI-LQ}. (See also \fref{fn:strict}.)
\\[1mm]
(ii) Calling the classical fields $A$ and $\Phi$ (suggesting a vector
potential and a scalar field) is just a matter of convenience that 
reflects their nature in the examples below. According to the
complexity of a model, a generalization with different fields
may be needed, and is possible.
\\[1mm]
(iii) It is legitimate to insert quantum fields into a classical
  Lagrangian, because the fields in a Lagrangian do not satisfy any
  equations of motion, while the Wick algebra is a quotient of the
  classical field algebra by the free equations of motion. A classical
  Lagrangian is rather used to find the equations of motion by
  Hamilton's principle, and its free part is used to ``canonically
  quantize'' classical fields. In the present setup, neither Hamilton's
  principle is imposed, nor is there need of canonical quantization,
  because the free quantum fields are directly constructed from the
  Wigner representations of the particles. In particular, the symbol
  $\,\wh\cdot\,$ does not stand for ``canonical quantization''. In sQFT,
  $\wh A$ will be (a linear function of) the string-localized potential $A(c)$,
  and $\wh\Phi$ has to be determined in each case, see \eref{hidden}
  and the examples in \sref{s:SM}. 
\end{remk}
\begin{prop}\label{p:HGI-LQ}
  When $L\ww$ and $Q^\mu$ enjoy the properties (P1)--(P4), the
  all-orders condition \eref{ILQ} for string-independence of $S_{L\ww}$
  at tree level is automatically fulfilled. 
\end{prop}
\emph{Proof:} By (P1) and \cref{c:ILQ}(i),  \eref{ILQ} is
equivalent to the condition \eref{ILQ0}. On the other hand,
$$(\ddelta+\omega_Q)(L_0+L\ww) \stackrel{\mathrm{(P2)}}=
(\ddelta+\omega_Q)(L[\wh A,\wh\Phi])\stackrel{\mathrm{(P4)}}=
\delta_{\lambda}(L[\wh A,\wh\Phi]) \stackrel{\mathrm{(P3)}}=0.$$
In the last step, it was used, that gauge invariance of $L[A,\Phi]$ makes no
specific assumptions about the nature of the fields and the gauge
parameters. Thus, condition \eref{ILQ0} is fulfilled. \qed
\begin{remk}\label{rk:no-SSB}
  The fact, that (in the examples with MVBs in \sref{s:SM}) the field $\wh \Phi$
  has a nontrivial vacuum expectation value, does not mean that there
  is ``spontaneous symmetry breaking'', see \sref{s:intro}, Item~6.
  Namely, $\wh A$ and $\wh \Phi$
  have no physical meaning: their only role is to assist in the proof
  of the condition \eref{ILQ} for string-independence. Along with $(\wh 
  A,\wh \Phi)$, \emph{every} gauge transform of $(\wh A,\wh \Phi)$ --
  with a different, even $x$-dependent vacuum expectation value --
  would satisfy \eref{hGI-LQ} as well (because $L[A,\Phi]$ is gauge-invariant).
  \end{remk}

  \paragraph{2. Equality of S-matrices.}
  We now show how hidden gauge invariance secures equality of S-matrices
  \eref{SLSK}, where $L_0$ from (P1) and $K\ww$ are functions of
  $\ddelta$-invariant fields only. (The latter is actually not needed --
  it only fixes the idea that the interaction $K\ww$ to compare with
  is manifestly $\ddelta$-invariant.) In sQFT, this means that $L_0$
  and $K\ww$ are functions of \emph{local} fields only.   

We have to assume that the fields involved in the interactions
$L\ww$ and $K\ww$ are defined on the same Fock space, so that
\eref{SLSK} is meaningful. In theories with \emph{massless} vector
bosons, the Hilbert space of $L\ww$ must be embedded into the
indefinite state space of the gauge theory interaction $K\ww$. This
requires a minor modification of \pref{p:ME} due the appearance of
null fields, cf.\ \cite[Sect.~4]{sQCD}. \pref{p:HGI-LV} below can be  
adapted accordingly. Further subtleties in the massless case arise due
to a logarithmic IR divergence. These are well-understood with
  important physical consequences in the case of QED, as will be briefly sketched 
in \sref{s:SM}, Item~1. Analogous results for QCD can be derived
  \emph{formally}, but with caveats addressed in \rref{rk:LV-YM}. We
ignore the IR issue here
and present the argument in a way which is litterally applicable
only when there are no massless vector bosons. 

On top of (P1)--(P2) (determining $L_0$ and $L[A,\Phi]$), we assume
(again at tree level):
\begin{itemize} \itemsep-0.5mm
\item[(P5)]  There exists a mediator field $U^\mu$ such that $\pa_\mu
  U^\mu =-\omega_U(L_0)$. 
\item[(P6)] There exist ($\ddelta$-invariant) quantum fields $A_0$ and
  $\Phi_0$ such that $L_0+K\ww$ can be written as
  \bea{hGI-LV}
  (L_0+K\ww)(x) =L[A_0,\Phi_0](x).
  \eea 
\item[(P7)] $L[A,\Phi]$ is invariant under arbitrary \emph{finite} gauge
transformations $\alpha_\gamma$.   
\item[(P8)] There exist Wick-algebra-valued, field-dependent gauge
  parameters $\gamma$ such that 
  $$e^{\omega_U}(\wh A) = \alpha_\gamma(A_0), \quad e^{\omega_U}(\wh
  \Phi) = \alpha_\gamma(\Phi_0).$$
  \end{itemize}

Again, these strong assumptions are only justified because they are
satisfied for the abelian Higgs model, and with appropriate modifications
also for models with massless vector bosons (YM and electroweak).

\begin{prop}\label{p:HGI-LV}
  When $L\ww$, $K\ww$ and $U^\mu$ enjoy the properties
  (P1)--(P3) and (P5)--(P8), the condition \eref{ILQ} for coinciding S-matrices
  $S_{L\ww}=S_{K\ww}$ at tree level is automatically fulfilled.
\end{prop}
\emph{Proof:} By (P5) and \cref{c:ILQ}(ii),  \eref{ME} is equivalent to
the condition \eref{ME0}. On the other hand, the remaining assumptions
and the fact that $e^{\omega_U}$ is a Wick algebra homomorphism,
allow to conclude 
\bea{ILQ-LV}
e^{\omega_U}(L_0+L\ww) &\stackrel{\mathrm{(P2)}} =&
e^{\omega_U}\big(L[\wh A,\wh \Phi]\big)
= L\big[e^{\omega_U}(\wh A), e^{\omega_U}(\wh \Phi)\big] \stackrel{\mathrm{(P8)}}=
L(\alpha_\gamma(A_0),\alpha_\gamma(\Phi_0)) \notag \\
&\stackrel{\mathrm{(P7)}}=& L[A_0,\Phi_0] \stackrel{\mathrm{(P6)}}=
L_0+K\ww.
\eea
This is \eref{ME0}. Then \cref{c:ILQ}(ii) implies the equality of S-matrices. \qed

\section{Standard model interactions}
\label{s:SM}
All Standard Model interactions involve massless particles of helicity
$\pm 1$ (photons, gluons) and/or massive vector bosons of spin 1 (MVBs:
$W$ and $Z$-particles). 

Various (sub-)models have been treated previously in the
  autonomous approach. We shall not represent detailed computations
  (except pointing to some changes due to to the new ``integrated'' formalism as
  mentioned in \sref{s:obst}, Item~2), but rather focus attention to
  the aspect of hidden gauge 
  invariance.  We begin with QED, YM and QCD.
We then turn to the abelian Higgs model (one MVB, no photon), which is
instructive for the more general case, treated last: the electroweak interactions.

In all models, the following ``logical structure'' is found (as
  indicated after \eref{ILQ3}): The particle content selects first-order
  interactions $L_1(c)$ satisfying the ``initial condition'' \eref{LQ1}.
  The analysis of second-order obstructions against string-independence
  of $\iint T[L_1(x,c)L_1(x',c)]$ fixes some free parameters and induces the
  second-order interactions $L_2(c)$ via \eref {ILQ2}. While the pedestrian
  step to third order via \eref {ILQ3} typically becomes very contrived (see, e.g.,
  \cite{chir,weak}), hidden gauge invariance emerging at second order
  can be used to conclude that the constraints are solved at all
  orders with $L_n=0$ ($n\geq 3$).  

\paragraph{1. QED and massive QED. \cite{infra,MR}}
The easiest interaction satisfying the initial first-order condition
\eref{LQ1} is
$$L_1=A_\mu(c)\, j^\mu, \quad Q^\mu= w\,j^\mu,$$
where $j$ is a conserved local current, say $j^\mu=\ol\psi \gamma^\mu
\psi$. The string-localized vector potential $A_\mu(c)$ may be
massless or massive. With kinematic propagators, one find
$\omega_Q(j)=\omega_Q(L_1)=\omega_Q(Q)=0$. Therefore, \eref{ILQ} is
readily solved with $L\ww(c)=g\, L_1(c)$ and $Q^\mu=g\, Q_1$. Property
(P1) is \emph{not} fulfilled because the only candidate for $L_0$
would be the Dirac Lagrangian that is identically zero as an
operator. (This seems to be a general feature for couplings to Fermi
fields.) Yet, string-independence follows directly with \pref{p:ILQ}.

In the massive case, also \eref{LV1} holds with
  $K_1=B_\mu\,j^\mu$, where $B_\mu$ is the massive Proca potential,
  and $U_1^\mu=\psi(c)\, j^\mu$. The tree-level equality of
  S-matrices for $L\ww$ and $K\ww = g\, K_1$ follows immediately with
  the mediator field $U=g\,U_1$,
  because with $\omega_U(L_1)=\omega_U(U_1)=\omega_U(K_1)=0$,
  \eref{ME} reduces to $L\ww -K\ww = \pa_\mu U^\mu$. With
  $\omega_U(\psi) = ig\phi(c)\,\psi$, the dressed Dirac
  field becomes string-localized:
  $$e^{\omega_U}(\psi(x))= e^{ig
  \phi(x,c)}\cdot \psi(x).$$
(The ongoing
  analysis \cite{MR} indicates the preservation of these results
  also at loop level.)  

In the massless case (the actual QED), the same relations are true,
but the field $\phi(c):=I^\nu_c(A_\nu)$ is problematic: While, by a
suitable choice of the string
  function $c$, it lives on a positive-definite subspace of the indefinite
  Fock space of the gauge vector potential $A$, its correlation functions
  are \emph{logarithmically divergent}. Fortunately, this IR divergence
  is manageable \cite{infra} because $\phi(c)$ appears only in exponentials:
  $e^{ig\phi(c)}$ can be non-perturbatively constructed in the GNS
  representation of the limit of states on massive Weyl algebras, as $m\to0$.
  There are several important physical consequences  \cite{infra}:

First, the vacuum state on the Weyl algebra of the exponential factors
  turns out to be a non-regular state, whose GNS representation splits into
  uncountably many superselection sectors labelled by the azimuthal
  profile of the function $c$. Second, longitudinal degrees of freedom
  contained in $\phi(c)$ are transferred by the dressing to the charged field: 
  the exponential factor of the latter describes a ``photon cloud'' (it creates
  a non-regular ``coherent photon state'') producing a nontrivial electric flux at 
  infinity depending on the superselected profile of $c$. It accounts
  for Gauss' Law by the normalization \eref{pac}.
  \footnote{\label{} The Coulomb dressing factor of \cite[Eq.~(18,19)]{Di} and
    \cite[Eq.~(4.3)]{LM} is the special case where the string function
    is $c^0=0$ and $c^i(y)=\delta(y_0)\pa_i\Delta\inv(y)$. With this
    choice, $e^{ig \phi(c)}$ is completely delocalized.} 
Third, the photon cloud contributes to the energy of the electron, which
  consequently is no longer a sharp mass eigenstate (``infraparticle'').
  These results, that cannot be obtained
with local perturbation theory,  are in accord with previous axiomatic
results for QED \cite{Bu,FMS1}. The non-locality of the dressed Dirac field
also supports the explanation of the Aharonov-Bohm effect by the
electromagnetic field of the electron including longitudinal photons
(e.g., \cite{Sa,Ear2}), rather than a proof that gauge potentials are
``real''. (There is no conflict with Einstein causality, because the
photon cloud was created ``in the past''.)

\paragraph{2. Yang-Mills. \cite{GGM,sQCD}}
Autonomous interaction densities $L\ww(c)$ for self-interactions of
any finite number of massless particles of helicity $\pm1$ were classified in
\cite{GGM}. String-independence requires that the vector fields
$A_a(c)$ and their field tensors $F_a$ form adjoint multiplets
of a real Lie algebra $\gg$ with completely antisymmetric structure
constants $f_{abc}$.
\footnote{\label{fn:tau}This means that $\gg$ is the Lie algebra of a
  subgroup of some unitary matrix group. The generators $\tau_a$ are
  anti-selfadjoint matrices. Our conventions for $\u(1)$ and for
  $\su(2)$ will be $i\tau_0=\eins$ and $i\tau_j=\frac 12 \sigma_j$,
  such that $f_{ijk}=\eps_{ijk}$. \\ 
  In \cite{sQCD}, we rather used a convention with self-adjoint generators
  $\tau'_a=-i\tau_a$ such that $i[\tau'_a,\tau'_b]=\sum_cf_{abc}\tau'_c$. \\
  The present reformulation in terms of adjoint field multiplets
  avoids the somewhat weird notion of ``Lie-algebra-valued quantum fields''
  $X = \sum_a \tau'_a X_a$ in \cite{sQCD}.}
For adjoint field multiplets, we write
$([X,Y])_a:= \sum_{ab}f_{abc}X_bY_c$, and $\erw{X\vert Y} :=
\sum_aX_aY_a$, hence $\erw{[X,Y]\vert Z}=\erw{X\vert[Y,Z]}$.

The autonomous $L\ww(c)$ was found to be
(Wick ordering is understood, as always)
\bea{L-YM}
L\ww(c) = \sfrac g2 \Erw{F^{\mu\nu}\Vert
  [A_\mu(c),A_\nu(c)]}-\sfrac{g^2}4 \Erw{[A^\mu(c),A^\nu(c)]\Vert
  [A_\mu(c),A_\nu(c)]}.
\eea
These are the cubic and quadratic terms of
the classical YM Lagrangian
\footnote{\label{s-sign}Our sign convention for the coupling constant
  is consistent with \cite{chir,weak,sQCD}, while $g$ is $-g$ in most
  modern textbooks.}
\bea{GA}
L[A]= -\sfrac14 \erw{G_{\mu\nu}[A]\vert
  G^{\mu\nu}[A] }, \qquad
G_{\mu\nu}[A]\equiv\pa_\mu A_\nu-\pa_\nu A_\mu - g [A_\mu,A_\nu],
\eea
with the classical field $A$ replaced by the quantum
field $A(c)$ in the Wick algebra. In other words, the ``hidden'' gauge
invariance (properties (P2) and (P3)) is quite manifest with $L_0 =
-\frac14\erw{F\vert F}$ and $\wh A=A(c)$. (A field $\Phi$ as in
\eref{hGI-LQ} is not needed for Yang-Mills.) We verify also
the properties (P1) and (P4). 

(P1) is easy: For any skew-inert field $\rho_\nu$ (i.e.,
$O_{[\mu}(\rho_{\nu]}(y),X(x))=0$ for all fields $X$), and
\bea{Qrho}
Q^\mu = \Erw{F^{\mu\nu}\Vert \rho_\nu},
\eea
it holds (\cite[Lemma~A.3]{sQCD})
\bea{omega-YM}
\omega_Q(A_\mu(c)) = \rho_\mu+ \pa_\mu
I_c^\nu(\rho_\nu), \quad \omega_Q(F_{\mu\nu}) = \pa_{[\mu}\rho_{\nu]} .
\eea
Because $\delta_c(F)=0$ and $\pa_\mu F^{\mu\nu}=0$, (P1) follows: 
\bea{dQ-YM}
(\delta_c+\omega_Q)(L_0) = -\Erw{F^{\mu\nu}\Vert
  \pa_\mu\rho_\nu}=-\pa_\mu Q^\mu.
\eea
Of course, $\rho_\nu$ must be further specified so that $Q_\mu$
fulfills also (P4). This is done by 
\begin{lemma}\label{l:sGT} 
  Let $\lambda= \sum_{n}g^n\lambda_n$ and
  $\rho_\nu$ (both adjoint field multiplets) be power series in the
  coupling constant with Wick-algebra-valued coefficients, defined by
  \bea{lambdarho}
  \lambda_0=w, \qquad \lambda_{n+1}=
I_c([\lambda_n,A(c)]), \qquad \rho_\nu := g[\lambda,A_\nu(c)].
\eea
$\lambda$ is an inert field in the sense of \cite[Sect.~3.5]{sQCD}
(i.e., $O_\mu(\lambda_n(y),X(x))=0$ for all $X$ and all $n$), and
$\rho_\nu$ is skew-inert. The  derivation $\delta_c+\omega_Q$ with $Q$
given by \eref{Qrho}, acts on $A(c)$ and $F$ like 
\bea{sGT}
(\delta_c+\omega_Q) (A_\mu(c,x)) &=&\delta_\lambda(A_\mu(c,x))
:= \pa_\mu \lambda(x) + g [\lambda(x),A_\mu(c,x)],
\\ \notag
(\delta_c+\omega_Q)(F_{\mu\nu}(x)) &=& g \,\pa_\mu
[\lambda,A_\nu(c,x)] - ({\mu\leftrightarrow\nu}).
\eea
The map $\delta_\lambda$ (extended to the field algebra as a derivation)
has the same form as the usual infinitesimal classical gauge
transformation of a classical field $A$ with gauge parameters
$\lambda$.
\footnote{\label{}However, it differs from the gauge
  transformation in \cite{sQCD}, whose
  action on the string-localized potential is induced
  from the gauge transformation
  $\delta_\lambda(A)=\pa\lambda+g[\lambda,A]$ of the \emph{local}
  gauge potential $A$ via $A(c)=A+\pa I_c(A)$.}
\end{lemma}
\emph{Proof:} The proof will show that $\lambda$ given
by \eref{lambdarho} is unique to satisfy \eref{sGT},
while $\rho_\nu$ is unique only up to a derivative.

$\lambda$ is inert by the same argument for inertness of $\gamma$ in
\cite{sQCD}; namely, it is legitimate to choose kinematic propagators for the 
expansion coefficients $\lambda_{n}$ because they are not subject to
equations of motion. Then, the skew-inertness of $\rho_\nu$
follows as in \cite[Lemma~A.2]{sQCD}. 

With \eref{deltac} and \eref{omega-YM}, the claim \eref{sGT} is the
assertion that  
\bea{pawGT}
\delta_c(A_\mu(c)) + \omega_Q(A_\mu(c)) = \pa_\mu w + (\rho_\mu + \pa_\mu I_c(\rho))
&\stackrel!=& \pa_\mu\lambda + g[\lambda,A_\mu(c)],
\eea
which implies the second equation in \eref{sGT} because both $\delta_c$
and $\omega_Q$ respect the exterior derivative, see \eref{omega-YM}.
Applying $-I_c^\mu$ to \eref{pawGT} and using \eref{Ipa}, gives the
implicit equation for $\lambda$ 
\bea{wGT}
\lambda -gI_c([\lambda,A(c)]) \stackrel!= w.
\eea
\eref{wGT} is solved by the power series $\lambda$ in \eref{lambdarho}.
Plugging \eref{wGT} into \eref{pawGT}, yields
\bea{rhoGT}
(\delta_\mu^\nu + \pa_\mu I_c^\nu)(\rho_\nu) = (\delta_\mu^\nu +
\pa_\mu I_c^\nu)(g[\lambda,A_\nu(c)]),
\eea
which is obviously solved by $\rho_\nu$ in \eref{lambdarho}.
\qed

By \lref{l:sGT} also (P4) is fulfilled, and by \pref{p:HGI-LQ} the S-matrix
is automatically string-independent at all orders. For the equality of
S-matrices, see \rref{rk:LV-YM}.

\paragraph{3. QCD. \cite{GGM,sQCD}}
We add a minimal coupling $\wt L\ww(c)= g\erw{A_\mu(c)\vert j^\mu}$
to the Yang-Mills interaction, where $(j_a^\mu)_a$ is an adjoint
multiplet of conserved quark currents.
\footnote{\label{fn:curr}With a general ansatz for currents $j^\mu_a =
  \ol \psi \gamma^\mu T_a\psi$ coupled to massless vector bosons,
  string-independence at first and second order requires that the
  matrices $T_a$ must commute with the fermionic mass matrix (i.e., quark
  masses are color-independent) and satisfy $i[T_a,T_b]=\sum_cf_{abc}T_c$,
  hence $iT_a=\pi(\tau_a)$ are a representation of the Lie algebra,
  see \aref{a:chir}. Consequently, $j_a$ are conserved and form an adjoint multiplet.}
As in QED, property (P1) is \emph{not} fulfilled in this case.
Yet, one can establish string-independence directly with \pref{p:ILQ}.

\begin{prop} \label{p:sQCD}
  Let $\wt Q^\mu =g\erw{\lambda\vert j^\mu}$ with $\lambda$ as in
  \lref{l:sGT}. The tree-level condition \eref{ILQ} for QCD with interaction 
  $L\ww(c)+\wt L\ww(c)$ is solved by $Q+\wt Q$. Hence, the S-matrix
  $S_{L\ww(c)+\wt L\ww(c)}$ is string-independent at tree level. 
\end{prop}
The abelian case (QED, $\lambda=w$) is of course included.

\noindent \emph{Proof:} Because $L\ww(c)$ and $Q$ separately satisfy the
condition \eref{ILQ}, we have to consider only the additional
obstructions arising from the fermionic couplings: 
\bea{ILQ-QCD}
\delta_c(\wt L\ww) - \pa_\mu \wt Q^\mu + \omega_Q(\wt L\ww)
+ \omega_{\wt Q}(L\ww) + \omega_{\wt Q}(\wt L\ww) \stackrel!=0.
\eea
Using the inertness of $\lambda$, \eref{omega-YM}, and \eref{deltac}, 
as well as $\omega_{\erw{\lambda\vert j}}(\erw{A(c)\vert j}) =
- \erw{[\lambda,A(c)]\vert j}$ (\cite[App.~A.2]{sQCD}), we see that
\bea{lala}
g\Erw{\pa_\mu w -\pa_\mu\lambda + \rho_\mu + \pa_\mu I_c(\rho) 
 -g [\lambda,A_\mu(c)] \Vert j^\mu}\stackrel!=0
\eea
must vanish. This is true by \eref{pawGT}. \qed

\begin{remk} \label{rk:Qwla} 
The all-order condition \eref{ILQ} is solved, in all models
  considered in this paper, by
\bea{Qwla}
Q^\mu=g\sumno_a \lambda_a(g)\sfrac{\pa L_1(c)}{\pa A_{a,\mu}(c)}
\eea
with $\lambda(g)=\sum_ng^n\lambda_n$ given by the same recursion
  as in \lref{l:sGT}. (In
particular, $Q_1=\sum_aw_a\frac{\pa L_1(c)}{\pa A_{a,\mu}(c)}$ solves \eref{LQ1}.)
Indeed, if one were to prove \pref{p:sQCD} with an ansatz
$\wt Q = g\erw{\lambda'\vert j}$ with $\lambda'$ chosen independently of
$\rho=g[\lambda,A(c)]$ in \lref{l:sGT}, then the condition \eref{lala}
would fix $\lambda'=\lambda$.

We don't have a model-independent argument why \eref{Qwla}
  should hold.
\footnote{\label{fn:Qwla} For general arguments at first and second orders,
  see \cite[Lemma~B.1, B.2]{weak}. The argument for $Q_2$ has to
  be adapted to the integrated reformulation of the present paper:
  in particular, with kinematic propagators, the term $Q_2\vert_\delta$ in
  \cite{weak} is absent. (The prefactor of the sum in
  \cite[Eq.~(B.1)]{weak} is a typo. It should be $+2$.)}
However, its correctness in all models follows from the fact that its
use allows to show that all obstructions can be resolved.
  \end{remk}
  \begin{remk}\label{rk:LV-YM}
    The mediating equation \eref{ME} for the
  equality of S-matrices of $L\ww(c)$ and the gauge-theoretic
  (indefinite-metric) interaction $K\ww$ was solved in \cite{sQCD},
  both for Yang-Mills and QCD. A minor modification is necessary
  because in $K\ww$, $\pa_\mu  F^{\mu\nu}= -\pa^{\nu}N$ rather than
  zero, where $N=\pa^\mu A_\mu$ is a null field. Consequently,
    \eref{ME} holds only up to null fields, that are ineffective in
  physical states.

When the subtleties of null fields are properly taken into account,
hidden gauge invariance as in (P5)--(P8) still prevails, see
\cite[Prop.~4.2]{sQCD}: One can perturbatively construct a mediator field
$U^\mu$ and (at
  least \emph{formally}, see below for caveats) a Lie-group-valued
  string-localized operator $W(c)=e^{i\gamma(c)}$, where
  $\gamma(c)=gI_c^\mu (A_\mu)$ + nonabelian higher-order terms is
  Lie-algebra-valued functional of the gauge potential $A$,
  involving Lie-algebra commutators and string integrations
  $I_c^\mu$. 
Then, one finds (P6) and (P8) with $\wh A=A(c)$ and $A_0=A$, and
  $\wh\Phi=\Phi_0=\psi$:
\bea{dress-YM}
e^{\omega_U}(A_{\mu}(c)) = W(c)(A_{\mu} -ig\inv\pa_\mu)W\inv(c), 
\qquad e^{\omega_{U+\wt U}}(\psi) =\pi(W(c))\psi.\eea

The dressed string-localized fields are field-valued ``gauge-transforms''
  of the fields of the indefinite-metric local approach. $W(c)$ behaves like a
``smeared Wilson operator'' \cite[Prop.\ 4.6]{sQCD}: viewed as a 
($c$-dependent) functional of the gauge potential $A$, it transforms
under compactly supported gauge transformations with number-valued 
gauge-parameters like a parallel transporter along any path from $x$ to $\infty$:
$$W^{[{}^\theta\! A]}(x,c) = W^{[A]}(x,c)\theta(x)\inv, $$
while $\gamma(c)$ and $W(c)$ are localized in a cone spanned by the
support of the string function $c^\mu$.
\footnote{\label{} For $c^\mu$ supported along a curve $C$ (which is too
  singular in QFT, though), $W(c)$ \emph{is} the parallel transporter
  along $C$.} Consequently, the dressed fields \eref{dress-YM} are
string-localized, and at the same time gauge invariant under local and
colored (adjoint resp.\ vector) under global gauge transformations.
This is possible because $\gamma(c)$ is field-valued. Thus, $W(c)$ is
the quantum realization of the inverse of a classical ``dressing field'' $u$
of QCD in the ``dressed field method'' in \cite{Fra}, cf.\ \rref{rk:DFM},
and also of the gauge-fixing dressing factor $h$ in \cite{LM}.   

The gauge-fixing dressing factor $h$ in \cite{LM} is
  \emph{postulated} to transform like
  $h[{}^\theta\!A]=h[A]\theta\inv$. Its formal non-perturbative construction 
  by ``functional integrations over the gauge group'' as in \cite{HILLM}
  breaks down in the nonabelian case because of 
  the Gribov problem. We emphasize, however, that our motivation for
  dressing is not a gauge fixing, but $W(c)$ is perturbatively
  ``generated'' by the obstructions, see \rref{rk:dress}. Also, it
  depends on a string function $c$ as an 
  input, rather than a gauge-fixing functional.

On the other hand, like $\phi(c)=I_c^\mu (A_\mu)$ in QED,
  $\gamma(c)$ defined on a positive-definite subspace of the
  indefinite Fock space of $A$ is logarithmically IR divergent, and
  its exponential requires a non-perturbative definition. But an analogous 
  construction of $W(c)$ in the nonabelian case is not known, because
  there is no Weyl algebra that can be used. Given that $W(c)$ and
  with $h$ are conceptually different things, one may doubt that their
  respective non-perturbative problems could be related; but one
  should certainly not expect that the Gribov problem could be visible 
  perturbatively, or even at tree level, in our dressed quark fields.

Note in this context  that our perturbative dressed quark
  fields, just like the gauge-fixed ones in \cite{LM} carry 
  purely gluonic clouds, while an interacting charged field -- if
  it exists -- should carry a ``cloud of gluons and sea
  quarks''. \eref{magic} explains why this is so: dressed fields must
  still be perturbed with the local (and localization-preserving) QCD
  interaction, which creates sea quarks.
\end{remk}

\paragraph{4. The abelian Higgs model. \cite{AHM}}
We present the abelian Higgs model (AHM) as a simple prototype for
nontrivial hidden gauge invariance. The general structure is the same
in the electroweak interaction (Item~4 below), but the computations
are most transparent in the AHM.

Free massive vector bosons are described by the local Proca field $B_\mu$
satisfying the Klein-Gordon equation with mass $m$ and
\bea{eom-Proca}
\pa_\mu B^\mu=0, \quad \pa_\mu F^{\mu\nu}=-m^2 B_\mu, \quad
F_{\mu\nu}=\pa_{[\mu}B_{\nu]}.
\eea
The string-localized fields are $A_\mu(c)$ and $\phi(c)$ as in
\sref{s:intro} with equations of motion
\bea{eom-AHM}
\pa_{[\mu}A_{\nu]}(c)=F_{\mu\nu}, \quad \pa_\mu A^\mu(c) = -m^2 \phi(c),
\quad \pa_\mu\phi(c) = A_\mu(c)-B_\mu.
\eea
The free Higgs field is a canonical scalar field $H$ of mass $m_H$.

$B$, $F$, and $H$ are string-independent local fields, while the
string variations of $A(c)$ and $\phi(c)$ were given in \eref{deltac}.

The autonomous interaction density $L\ww(c)$ for the interaction of
a single massive vector boson with a scalar was determined in
\cite{Schroer,AHM}. The improvements of the present paper,
outlined in \sref{s:obst}, Item~2, change the computation in several
ways. First, the integration in \eref{omega} deletes derivatives 
in the variable $y$, while they survived in \cite{AHM} due to the
symmetrization in $x$ and $y$. Thus, $Q_2$ in \cite{AHM} is now absent.
Because the presence of $Q_2$ in \cite{AHM} prevented
the resolution of higher-order obstructions, the authors were forced to
adopt renormalizations of  propagators $\erw{0\vert T[BB]\vert 0}$
and
$\erw{0\vert T[\pa H\pa H]\vert 0}$ ($c_B=c_H=-1$). These are now obsolete.
Instead, we work with kinematic propagators ($c_B=c_H=0$). As a
consequence, some extra terms $L_2^*(c)$ must be included into the
second-order interaction $L_2(c)$. They can also be found in
\cite[Eq.~(3.14)]{AHM}. 

To conclude: The cubic interaction $L_1(c)$ and $Q^\mu=gQ^\mu_1$
are the unique solution to the initial condition \eref{LQ1} with the
given field content, and the quartic part $L_2(c)$ (including
$L_2^*(c)$) is determined by string-independence at second order
\eref{ILQ2} with kinematic propagators. $Q^\mu$ satisfies \eref{Qwla}
in the abelian case $\lambda=w$.
\begin{prop}\label{p:AHM-LQ} The abelian Higgs model with $L_0$,
  $L\ww(c)$, and $Q$ given by
\footnote{\label{FF}The term $-\frac14 FF$ is not needed because it is
separately invariant under $\delta_c$ and $\omega_Q$, and separately
gauge invariant. We include it here for the sake of congruence
with nonabelian models, where it is not invariant under $\omega_Q$.}
  \bea{Ltot-AHM}
L_0 &=& -\sfrac14 F_{\mu\nu}F^{\mu\nu}+\sfrac12 m^2 B_\mu B^\mu +
\sfrac12 \pa_\mu H\pa^\mu H - \sfrac12 m_H^2 H^2, \notag \\ 
L_0+L\ww(c) &=& -\sfrac14 F_{\mu\nu}F^{\mu\nu} + \sfrac12
m^2\big(B+\sfrac gm A(c)H\big)^2
+ \sfrac12 \big(\pa H + mg A(c)\phi(c)\big)^2 - V(H,\phi(c)), \notag
\\ &&\hbox{where} \quad V(H,\phi) =
\sfrac12m_H^2\big(H+\sfrac g{2m}(H^2+m^2\phi^2)\big)^2, \\ \notag
Q^\mu &=& mg\,w(B^\mu H+\phi(c)\pa^\mu H)
\eea
enjoys the properties (P1)--(P4), i.e., the tree-level S-matrix is
string-independent at all orders. 
\end{prop}
\emph{Proof:} $L_0$ is determined by $Q^\mu$ ``in the autonomous
way'' as in \cref{c:ILQ}(i) and (P1), i.e., via $\pa_\mu Q^\mu = -
(\delta_c+\omega_Q)(L_0)$, as follows: The obstruction map
$\omega_Q(\cdot)$ is computed by \eref{factorwise} from the two-point 
obstructions given in Table~A.1. This gives Table~4.1. 
\begin{table}[h] 
$$\begin{tabular}{@{}l||c|c|c|c|c|c|@{}}
& $F$ & $A(c)$ & $B$ & $\phi(c)$ & $H$ &  $\pa H$
\\ \hline\hline
$ \omega_Q(\cdot)\quad$  &  $0$  &  $0$  & $-\frac gm \pa (wH)$
 &  $\frac gm \, wH$ & $-mg\,w\phi$
&  $-mg\, \pa(w\phi)$ 
\mystrut{11}{6} \\ \hline 
\end{tabular}$$
  \caption{The obstruction map $\omega_Q$ in the abelian Higgs model.}
  \label{tb:omQ-AHM} 
\end{table}

Then, the equations of motion and Table~4.1  immediately imply
$\pa_\mu Q^\mu = - \omega_Q(L_0)$ where $L_0$ given in
\eref{Ltot-AHM}. 
Because $\delta_c(L_0)=0$, (P1) is fulfilled.

Notice that $L_0$ is
the sum of the free Proca and scalar Lagrangians. Next, let
\bea{LAPhi-AHM} L[A,\Phi] = -\sfrac14\Erw{F_{\mu\nu}[A]\Vert F^{\mu\nu}[A]} +
(D_\mu\Phi)^*D^\mu\Phi - V(\Phi^*\Phi),
\eea
where $A$ and $\Phi$ are a vector field and a complex scalar field
with covariant derivative  $D_\mu\Phi = (\pa_\mu-ig A_\mu)\Phi$, and
$V$ is the double-well potential:
\bea{V}
V(\Phi^*\Phi) = \sfrac{m_H^2}{8v^2}(\Phi^*\Phi-v^2)^2.
\eea
Then, (P2) and (P3) are fulfilled with
\bea{AcPhic-AHM}
\wh A_\mu := A_\mu(c), \quad \wh \Phi := v+H + im\phi(c),
\eea
provided the parameter $v$ in \eref{LAPhi-AHM} is chosen to be
\bea{}
v=\sfrac mg.
\eea
Clearly $L[A,\Phi]$ is invariant under infinitesimal and finite gauge
transformations
\bea{GT-ab}
\delta_\lambda (A) = \pa \lambda, \quad \delta_\lambda(\Phi)
= ig\lambda \,\Phi,  \quad\hbox{resp.}\quad 
\alpha_\gamma (A) =  A+\pa \gamma, \quad
\alpha_\gamma(\Phi) = e^{ig\gamma}.
\eea
In the verification of (P2),  the crucial detail is the field equation
$\pa\phi(c)=A(c)-B$, see \eref{paphi}: 
\bea{}
\notag (\pa - ig \wh A)\wh \Phi &=& \pa H + im \pa\phi(c) -ig A(c)
(v+H+im\phi(c)) \\ \notag &=& \pa H + mg A(c)\phi(c)- i\big(m B + g A(c) H
\big).
\eea
This, together with $F[\wh A] = F$ and $\wh\Phi^*\wh\Phi =
(v+H)^2+m^2\phi(c)^2$, yields \eref{hGI-LQ}. 

\newpage

There remains (P4). From Table~4.1, we read off the action of
$(\delta_c+\omega_Q)(\cdot)$. We compute
\bea{}
(\delta_c+\omega_Q)(\wh A_\mu) = \pa_\mu w, \quad
(\delta_c+\omega_Q)(\wh\Phi) = - mg\,w\phi(c) +im(w+\sfrac gm\,w H)
=igw \,\wh\Phi. \notag
\eea
In view of \eref{GT-ab}, (P4) is fulfilled with $\lambda=w$.

With (P1)--(P4) fulfilled, the S-matrix is
string-independent at all orders by \pref{p:HGI-LQ}. \qed

Recall \rref{rk:no-SSB} that the vacuum expectation
$\erw{0\vert \wh\Phi\vert 0}=v$ does not mean that there is 
spontaneous symmetry breaking. 

\begin{prop}\label{p:AHM-LV}
  The abelian Higgs model also enjoys the properties (P5)--(P8) with
  \footnote{\label{reno} The couplings of $\frac{m^2}2B^2$ to $(1+\frac gm
    H)^2-1$ in \eref{K-AHM} and to $1-(1+\frac gm
    H)^{-2}$ in \cite{AHM} are related by a renormalization group
    transformation taking $c_B=0$ to $c_B=-1$. The argument in
    \cite[Sect.~4.7]{AHM} can be
  extended to all orders. }
\bea{K-AHM}
K\ww &=& mgB_\mu B^\mu (H + \sfrac g{2m}H^2) - \sfrac12m_H^2(\sfrac
g{m}H^3+\sfrac{g^2}{4m^2}H^4), 
\\ \notag \hbox{hence} \quad
L_0+K\ww &=& -\sfrac14 F_{\mu\nu}F^{\mu\nu}
+ \sfrac12 m^2 B^2(1+\sfrac gmH)^2 +\sfrac12 (\pa H)^2 - V(H,0),
\eea
hence $S_{L\ww(c)}=S_{K\ww}$ at tree level. 
The mediator field $U^\mu$ is of the form
\bea{U-AHM}
U^\mu(c) = \beta(H,\phi(c)) \cdot B^\mu + \eta(H,\phi(c))\cdot \pa^\mu H,
\eea
where $\beta$ and $\eta$ have to be determined as power series,
and $A_0:=B$ and $\Phi_0:=v+H$.
\end{prop}
\emph{Proof:}
We do not need to know the functions $\beta(H,\phi)$ and $\eta(H,\phi)$
to establish (P5). Namely, \eref{U-AHM} together with Table~A.1
and the inertness of $\phi$ and $H$ yields Table~4.2.
\begin{table}[h] 
$$
\begin{tabular}{@{}l||c|c|c|c|c|@{}}
& $A(c)$ & $B$ & $\phi(c)$ & $H$ &  $\pa H$
\\ \hline\hline
 $\omega_U(\cdot)\quad$&$0$& $-m^{-2}\pa \beta(H,\phi(c))$
& $m^{-2}\beta(H,\phi(c))$  & $-\eta(H,\phi(c))$
 & $-\pa\eta(H,\phi(c))$
\mystrut{11}{6} \\ \hline 
\end{tabular}
$$
\caption{The obstruction map $\omega_U$ in the abelian Higgs model.}
\label{tb:omU-AHM} 
\end{table}

Then, the equations of motion and Table~4.2 immediately imply (P5):
$$\pa_\mu U^\mu = \pa_\mu \beta \cdot B^\mu +\pa_\mu\eta \cdot \pa^\mu H  -
\eta\cdot m_H^2H = -\omega_U(L_0) .$$
In the verification of (P6), the crucial part is again the covariant derivative: 
$$(\pa_\mu-igA_{0,\mu})\Phi_0= (\pa_\mu-igB_\mu)(v+H) = \pa_\mu H -
imB_\mu(1+\sfrac gmH).$$
(P7) is obviously satisfied. For the nontrivial property (P8), we need a Lemma.
\begin{lemma}\label{l:dress-AHM}
  There exist unique real power series (in $g$) $\beta(H,\phi)
  = g H\phi + \dots$ and $\eta(H,\phi) = \frac g2\phi^2+\dots$, such
  that $\omega_U$ with $U=\beta \,B+\eta \,\pa H$ as in \eref{U-AHM}
  induces at tree level the transformation
\footnote{\label{}
The real and imaginary parts of \eref{dress-AHM} give the dressings of
$H$ and $\phi(c)$ separately. They were first conjectured in \cite{LV} by the
perturbative evaluation of string-independence until
$g^5$. \pref{p:AHM-LV} shows the conjecture to be correct at all orders.
\\ Observe that \eref{dress-AHM} takes complex polar to
Cartesian coordinates. It could be an appealing challenge to analytically
understand the interpolating flow $e^{t\omega_U}$ between polar and
Cartesian, and the field
$U^\mu$ as its generator.
}
\bea{dress-AHM}
e^{\omega_U}(\wh\Phi) = e^{\omega_U}(v+H+im\phi(c)) =
e^{ig\phi(c)}(v+H) = e^{ig\phi(c)}\Phi_0.
\eea
\end{lemma}
\emph{Proof:}
In view of Table~4.2 and $v=\frac mg$, \eref{dress-AHM}
becomes the differential equation
\bea{U-diff}
e^{m^{-2}\beta(H,\phi) \pa_\phi -\eta(H,\phi)\pa_H}(1+\sfrac gmH + ig\phi)
= e^{ig\phi}(1+\sfrac gmH).
\eea
At each order $g^{n+1}$, the left-hand side  is 
$(\frac 1{m^2}\beta_n\pa_\phi -\eta_n\pa_H)(\frac1m H+i\phi) \equiv
- \frac1m\eta_n+\frac{i}{m^2}\beta_n$ plus terms involving $\beta_\nu$ and
$\eta_\nu$ ($\nu<n$). This allows to recursively solve \eref{U-diff} for
$\eta_{n}$ and $\beta_n$. \qed 

Now, on top of \eref{dress-AHM}, $\omega_U(A(c))=0$ implies
$e^{\omega_U}(\wh A) = A(c)=B+\pa \phi(c)=A_0+\pa \phi(c)$. Clearly,
the right-hand sides are gauge transforms of  $\Phi_0$ and $A_0$ with
$\lambda = \phi(c)$, that is, (P8). 

With (P5)--(P8) fulfilled, the tree-level S-matrices of $L\ww(c)$ and of $K\ww$
automatically coincide at all orders by
\pref{p:HGI-LV}. This concludes the proof of \pref{p:AHM-LV}. \qed

Notice that $K\ww$ is a non-renormalizable interaction.
Remember our attitude, expressed in \sref{s:obst},
that the preservation of the properties (P1)--(P8) should be imposed as a
renormalization condition at loop level, that would fix infinitely
many renormalization parameters at all orders. 
\begin{remk}\label{rk:DFM}
  In (P2) and in (P6), the relation $L_0+L\ww(c)=L[\wh A,\wh \Phi]$ and
  the relation between $L[e^{\omega_U}(\wh A),e^{\omega_U}(\wh\Phi)]$
  and $L[A_0,\Phi_0] = L_0+K\ww$, which involves a gauge
  transformation  with operator-valued and string-localized
  parameters, are identities among \emph{quantum} fields in the Fock
  space. This situation should be contrasted with the discussion in \cite{Fra}.
  
  In \cite{Fra} it is emphasized that $L[A,\Phi]$, regarded as a
  \emph{classical Lagrangian} (i.e., $A$ and $\Phi$ do not satisfy any
  field equations) is just another way of writing the Lagrangian $(L_0+K\ww)[B,H]$,
  ``devoid of any gauge symmetry''. The passage is made by a
  \emph{local} field redefinition: the polar decomposition
  $\Phi(x)=e^{i\chi(x)}(v+H(x))$ defines $H(x)$, and 
  $B(x):=A(x)+\pa\chi(x)$. This feature is interpreted by saying that
  the gauge invariance of $L[A,\Phi]$ is ``artificial'' -- which raises
  the question \cite{Ly,Ear1} whether the spontaneous breakdown of
  an artificial symmetry can have a physically meaning?
  Notably, in \cite{FMS2,Hi,Ki} alternative treatments of
    massive vector bosons without spontaneous breakdown of gauge
    symmetry were discussed.

In contrast, \cite{Fra} regards the gauge invariance of electrodynamics
and Yang-Mills or QCD as ``substantial'' because it cannot be
erased by a local field redefinition. Yet, these theories can be reformulated
in terms of classical gauge-invariant ``dressed fields'' constructed with the
help of a non-local ``dressing field'' $u$. The quantum version of these
dressed fields are our $e^{\omega_U}(A(c))$ and (when fermions are coupled)
$e^{\omega_{U+\wt U}}(\psi) =W(c)\psi$ as in \eref{dress-YM},
cf.\ \cite{sQCD}. See also \rref{rk:dress} and \rref{rk:LV-YM}.  
\end{remk}

\paragraph{5. The electroweak interactions. \cite{chir,weak}}
The most general autonomous interaction density $L\ww(c)$ for the
self-interaction of any (finite) number of massless and massive vector
bosons and a single scalar Higgs particle was determined in
\cite{weak}, by imposing string-independence at first and second
order. At third order, some parameters in $L_1(c)$ and  $L_2(c)$ were
fixed, and $L_3(c)=0$. (Generalizations with several Higgs are
possible, but at least one Higgs \emph{must} must be present.)

Because fields $B_a(x)$ and $\phi_a(c,x)$ do not exist for photons ($m_a=0$), we
introduce indices $\ua$ running only over the MVBs. The fields are
therefore the local field tensors $F_{a,\mu\nu}$, the vector potentials
$A_{a,\mu}(c)$ and their variation fields $w_a$ for all vector bosons
in a mass eigenbasis, and the local Proca fields $B_{\ua,\mu}$ and
the fields $\phi_\ua(c)$ for the massive vector bosons (MVBs). 

As in Yang-Mills, the self-couplings of vector bosons are described by
completely antisymmetric structure constants $f_{abc}$ of a Lie
algebra $\gg$ of compact type. We use the same notations as in
QCD, Item~2,
also when the fields may be massive. But the masses lead to further
constraints: in particular, the adjoint action $\ad_{\tau_b}$ of the
``massless generators'' of $\gg$ must ``preserve the mass'', i.e., $m_b=0$ and
$f_{abc}\neq 0$ implies $m_a=m_c$.
As a consequence, the massless generators generate a Lie subalgebra
$\hh\subset\gg$, and $\ad_\hh$ preserves the subspace $\mm\subset\gg$
spanned by the ``massive generators'' $\tau_\ua$. 

String-independence at second order imposes more constraints, see
\cite[Eq.~(5.7)]{weak}. 

For the same reason as explained before \pref{p:AHM-LQ} in the abelian
Higgs model, one must include terms $L_2^*(c)$ into $L_2(c)$ when
working with the new condition \eref{resnew}. These were displayed in
\cite[Eq.~(B.4)]{weak}. But the resolvability of the third-order
obstruction has to be re-done. While this would be an annoying task to
do ``by hand'', hidden gauge invariance comes to assist. Namely, we
shall show that for the autonomous string-localized interaction
$L\ww(c) = gL_1(c)+ \frac 12{g^2}L_2(c)$, there exists
$Q^\mu$, such that the properties (P1)--(P4) are fulfilled, securing
the resolution of obstructions at all orders. In particular,
$L_n(c)=0$ for $n\geq3$.

The electroweak interaction is the unique solution to all constraints
with one photon and three MVBs (except for cases where one particle
decouples) \cite{weak}. We label the vector bosons in a mass
eigenbasis by indices $A,1,2,Z$, where $1$ and $2$ are coupled by the
photon: $f_{A12}\neq0$, hence $m_1=m_2=:m_W$. We parametrize the
structure constants such that
\footnote{\label{}\eref{fabc} 
  differs from the normalization convention in \cite{weak} by a factor
  of $2$.}
\bea{fabc}
f_{A12}=\sin\theta, \quad f_{12Z}=\cos\theta
\eea
with an angle $\theta$ to be determined. $f_{AZi}$ ($i=1,2$) must
vanish because $m_W\neq m_Z$, see above.

With this input, the first and second orders of the interaction $L\ww(c)$
and $Q^\mu$ are uniquely determined by imposing \eref{LQ1} and
\eref{ILQ2} with kinematic propagators. As it turns out, the inclusion
of $L_2^*(c)$ and the addition of $L_0$ determined ``in the autonomous
way'' via (P1) simplify the structure of $L_0+L\ww(c)$: it becomes a
``sum of squares'',  of which $L_0$, $gL_1(c)$, $\frac12g^2L_2(c)$
are the quadratic, cubic, and quartic terms, respectively:
 \bea{Ltot-EW}
 L_0+L\ww(c) &=&-\sfrac14 \Erw{G_{\mu\nu}(c)\Vert G^{\mu\nu}(c)} 
+\sfrac12\sumno_\ua m_\ua^2 \big(B _\ua + C_{\ua}[A(c)] \big)^2
\notag \\  &&
+\sfrac12\big(\pa H + C_H[A(c)]\big)^2  - \sfrac 12m_H^2\big(H +
\sfrac {gK}2(H^2 + \sumno_\ua m_\ua^2\phi_\ua(c)^2)\big)^2 ,
\qquad 
\eea
where 
\bea{Ca}
C_{\ua,\mu}[A(c)] =  \sfrac g{m_\ua^2}\sfrac{\pa L_1(c)}{\pa B^\mu_\ua}
&=& g\sumno_{b\uc}\sfrac {m_\uc}{m_\ua}\ga_{\ua
  b\uc}A_{b,\mu}(c)\phi_\uc(c) +  gK A_{\ua,\mu}(c) H, \notag \\
C_{H,\mu}[A(c)] = g\sfrac{\pa L_1(c)}{\pa (\pa^\mu H)}
&=& gK\sumno_\ua m_\ua^2A_{\ua,\mu}(c)\phi_\ua(c).
\eea 
The coefficients $\ga_{\ua b\uc}$ are only defined for massive $\ua$ and $\uc$:
\bea{ga}
\ga_{\ua b\uc} := \sfrac{m_\ua^2-m_b^2+m_\uc^2}{2m_\ua m_\uc}
\cdot f_{\ua b\uc} =-\ga_{\uc b\ua},\qquad\hbox{in particular}\quad
\ga_{\ua b\uc}=f_{\ua b\uc} \quad\hbox{for}\quad b=A.\quad
\eea
The analysis in \cite[Sect.~4]{weak} shows that string-independence
requires $m_Z\geq m_W$ and  
\bea{spec}
\sfrac{m_W}{m_Z}=\cos\theta \quad \hbox{and} \quad K=\sfrac1{2m_W}.
\eea
The masses of the particles fix all parameters in \eref{Ltot-EW},
\eref{Ca} (except the coupling constant).

The quadratic part $L_0$ of \eref{Ltot-EW} is seen to be
the sum of the free Maxwell Lagrangian for the photon,  free Proca
Lagrangians for the massive vector bosons, and the free Lagrangian of
the scalar Higgs particle. The specific form of the ``shift terms'' \eref{Ca} in 
\eref{Ltot-EW} was ``autonomously'' fixed by the initial first-order 
condition \eref{LQ1}, and already anticipates the hidden gauge
invariance, see \eref{igAPhic-EW}. 

To exhibit the hidden gauge invariance, it is convenient to define
\bea{ggv}
g':= \tan\theta\cdot g, \qquad v:= \sfrac{2m_W}g.
\eea 
\begin{prop}\label{p:EW-LQ}
Let $L_0$ and $L\ww(c)=gL_1(c)+\frac{g^2}2L_2(c)$ as specified by
  \eref{Ltot-EW}, and $Q^\mu$ as in \eref{Qwla}, see \rref{rk:Qwla}.
  $\lambda$ is given by the same formula \eref{lambdarho} as for
Yang-Mills, but with electroweak vector bosons rather than gluons.

Let $A_\mu=(A_{i,\mu})_{i=0,1,2,3}$ and $\Phi$ an $\su(2)$-doublet of
$\u(1)$-charge $\frac12$ with covariant derivative
\bea{DPhi-EW}
D^{[A]}_\mu\Phi = \Big(\pa_\mu + ig' A_{0,\mu}\cdot \sfrac12\eins_2 + ig
\sumno_{i=1,2,3}A_{i,\mu}\cdot \sfrac12\sigma_i)\Big)\Phi .
\eea
Let $V(\Phi^*\Phi)= \frac{m_H^2}{8v^2}(\Phi^*\Phi-v^2)^2$, and
\bea{L-EW}
L[A,\Phi] := -\sfrac14\Erw {G_{\mu\nu}[A]\Vert G^{\mu\nu}[A]}
+ \sfrac12 (D^{[A]}_\mu \Phi)^*D^{[A],\mu}\Phi - V(\Phi^*\Phi).
\eea
Then the electroweak interaction $L\ww(c)$ together with $Q$ 
enjoys the properties (P1)--(P4), where
\bea{AcPhic-EW}
(\wh A_i)_{i=0,1,2,3} &:=& \big(\cos\theta\!\cdot\!
A_A(c)-\sin\theta\!\cdot\! A_Z(c),\, A_1(c),\, A_2(c), \, 
\sin\theta\!\cdot\! A_A(c)+\cos\theta\!\cdot\! A_Z(c),\notag \\
\wh\Phi &:=& \bpm -m_W(i\phi_1(c) +\phi_2(c)) \\ v+H - im_Z \phi_Z(c)
\epm.
\eea
Consequently, the tree-level S-matrix is string-independent at all orders.
\end{prop}
\begin{table}[h] 
$$\begin{tabular}{@{}l||c|c|c|c|c|c|@{}}
& $F_a$ & $A_a(c)$ & $B_\ua$ & $\phi_\ua(c)$ & $H$ &  $\pa H$
\\ \hline\hline
  $\omega_Q(\cdot)\quad$ & $\pa\wedge\rho_a$ & $\rho_a + \pa I(\rho_a)$ &
$\rho_\ua - \frac1{m_\ua^2}\pa \big(\frac{\pa Q}{\pa B_\ua}\big)$& 
$I_c(\rho_\ua) + \frac1{m_\ua^2} \frac{\pa Q}{\pa B_\ua}$ &$-
\frac{\pa Q}{\pa(\pa H)}$&
$-\pa \big(\frac{\pa Q}{\pa(\pa H)}\big)$
\mystrut{12}{8} \\ \hline 
\end{tabular}$$
\caption{The obstruction map $\omega_Q$ for the electroweak
  interactions. Here, $\rho_\nu = g[\lambda,A_\nu(c)]$ and
  $\frac1{m_\ua^2}\frac{\pa Q}{\pa B_\ua}=g\sum_{b\uc}\frac
  {m_\uc}{m_\ua} \ga_{\ua b\uc}\lambda_b\phi_\uc +
  gK\lambda_\ua H \equiv C_\ua[\lambda]$ and $\frac{\pa Q}{\pa (\pa H)}
  = gK\sum_\ua m_\ua^2\lambda_\ua\phi_\ua\equiv C_H[\lambda]$.}
\label{tb:omQ-EW} 
\end{table}
\emph{Proof:} In order to verify (P1), we need the equations of motion and
Table~4.3, which follows from $Q$ given by \eref{Qwla} and Table~A.1.
Thus, we have to work out
\bea{}
\notag \pa Q + \omega_Q(L_0) = \sfrac{\pa Q}{\pa
  \lambda_a}\pa\lambda_a + \sfrac{\pa Q}{\pa F_\ua}(-m_\ua^2
B_\ua) + \sfrac{\pa Q}{\pa \phi_\ua(c)}(A_\ua(c)-B_\ua)
+ \sfrac{\pa Q}{\pa H}\pa H +\sfrac{\pa Q}{\pa (\pa H)}(-m_H^2 H)  \\ 
\notag
-\sfrac12\Erw{F\Vert \pa\wedge\rho} + m_\ua^2B_\ua
(\rho_\ua-\sfrac1{m_\ua^2}\pa \big(\sfrac{\pa Q}{\pa B_\ua}\big))
+ \pa H \big(-\pa \big(\sfrac{\pa Q}{\pa(\pa H)}\big)\big) - m_H^2 H
\big(-\sfrac{\pa Q}{\pa(\pa H)}\big).
\eea
There are no contributions to $\pa Q$ from $\sfrac{\pa Q}{\pa A}
\sim \eta_{\mu\nu} F^{\mu\nu}=0$ and from $\sfrac{\pa Q}{\pa B}$ because
$\pa B=0$. A straightforward computation yields $\pa Q +
\omega_Q(L_0)=0$, while $\delta_c(L_0)=0$. This proves (P1).

We insert into $L[A,\Phi]$ the string-localized quantum fields $\wh A$
and $\wh \Phi$. Because the passage from $(\wh A_i)_i$ to $(A_a(c))_a$
just amounts
to an orthogonal change of basis of the Lie algebra
$\u(1)\oplus\su(2)$, the term $-\frac14\Erw {G[A(c)]\Vert G[A(c)]}$
equals the first term in \eref{Ltot-EW}. Because $\wh\Phi^*\wh\Phi-v^2=
2v(H + \frac1{2v}(H^2 + \sum_\ua m_\ua^2 \phi_\ua(c)^2)$
and $\frac1{2v}=\frac g2K$, the term $-V(\wh\Phi^*\wh\Phi)$ equals the
last term in \eref{Ltot-EW}.

Working out the vector-valued doublet
$\sfrac i2\big(g'\wh A_0\eins_2+g\sum_i\wh
  A_i\sigma_i\big)\wh\Phi$
in \eref{DPhi-EW} with the above specifications, we find
\footnote{\label{}The first term arises from $\bpm 0\\[-1mm]v\epm$ in
  $\wh\Phi$. The ``strange'' mass factors in $\ga_{\ua b \uc}$
  \eref{ga} appearing in $C_\ua$ in the second term arise from the
  Weinberg rotation. E.g.,  the matrix $\frac{g'}2 A_0(c)\eins +
  \frac{g}2A_3(c)\sigma_3$ has the $1$-$1$-component $\frac g2[2\sin\theta
  A_A(c)+\frac{2\cos^2\theta-1}{\cos\theta}A_Z(c)]  =g[\sin\theta
  A_A(c)+\frac{2m_W^2-m_Z^2}{2m_W^2}\cos\theta A_Z(c)]=-g [\ga_{1A2}
  A_A(c)+\ga_{1Z2} A_Z(c)]$}
\bea{igAPhic-EW}
i\sfrac g2\big(\tan\theta \wh A_0\eins+\sumno_i \wh A_i\sigma_i\big)\wh\Phi
&=& \bpm m_W(iA_1(c)+A_2(c)) \\
-im_ZA_Z(c) \epm \notag \\ &+& \bpm m_W(i C_1[A(c)]+C_2[A(c)]) \\
C_H[A(c)] - i m_Z C_Z[A(c)]\epm. \quad
\eea 
Adding $\pa\wh\Phi$ to the first term and using \eref{paphi}, one sees
that the real and imaginary parts of the components of $D^{[\wh
  A]}\wh\Phi$ are $m_\ua (B_\ua+C_\ua[A(c)])$ and $\pa H+C_H[A(c)]$.
Thus, the term $\frac12 \big\vert D^{[\wh A]}\wh\Phi\big\vert^2$ equals
the second and third terms in \eref{Ltot-EW}. This proves (P2). (P3)
is evident.

In order to verify (P4), we have to add $\delta_c(\cdot)$ to the
entries of Table~4.3 and work out
\bea{}
\notag (\delta_c + \omega_Q)(\wh A_\mu) &\stackrel!=&
\delta_\lambda(\wh A_\mu) = \pa_\mu\lambda + g[\lambda,\wh A_\mu],
\\ 
(\delta_c + \omega_Q)(\wh\Phi) &\stackrel!=&
\delta_\lambda(\wh\Phi) 
=-\sfrac i2\big(g'\lambda_0\eins_2+g\sumno_i\lambda_i\sigma_i\big)\wh\Phi,
\eea 
where the second line corresponds to the 
$\u(1)$-charge $\frac12$ and the vector representation of $\su(2)$ as
in \pref{p:EW-LQ}, with
$\lambda_0=\cos\theta\cdot \lambda_A-\sin\theta\cdot \lambda_Z$ and 
$\lambda_3=\sin\theta\cdot \lambda_A+\cos\theta\cdot \lambda_Z$.

The first condition is the same as \eref{pawGT}, solved by
\eref{lambdarho}. The right-hand side of the second condition is the
same as \eref{igAPhic-EW} with $A_a(c)$ replaced by $\lambda_a$.
Thus, it remains to check
$$(\delta_c + \omega_Q)\bpm -m_W(i\phi_1(c) +\phi_2(c)) \\ v+H - im_Z
\phi_Z(c)\epm \stackrel != -\bpm m_W(i\lambda_1+\lambda_2) \\ 
-im_Z\lambda_Z \epm  - \bpm m_W(i C_1[\lambda]+C_2[\lambda]) \\
C_H[\lambda] - i m_Z C_Z[\lambda]\epm,$$
which is obviously satisfied by Table~4.3 and \eref{wGT}.
This proves (P4).

With (P1)--(P4) established, the S-matrix is string-independent
by \pref{p:HGI-LQ}. \qed 

One may as well verify (P5)--(P8) with $A_0=A$ the massless gauge
potential and $\Phi_0=\bpm 0\\[-2mm] v+H\epm$. The mediator field $U$
has a Yang-Mills part, such that $e^{\omega_U}(A(c)) = \alpha_{\gamma(c)}(A) =
e^{i\gamma(c)}(A-ig\inv\pa) e^{-i\gamma(c)}$ with
$i\gamma(c):=\sum_a\gamma_a(c)\tau_a$  as in \cite{sQCD}, and an
MVB part to be determined similarly as in \lref{l:dress-AHM}, such that  
$e^{\omega_U}(\Phi(c))= \pi(W(c))\Phi(c)$. Then \pref{p:HGI-LV} applies,
and consequently, $S_{L\ww(c)}=S_{K\ww}$.

\paragraph{6. Quark and lepton couplings: The chirality theorem. \cite{chir}}
Given $L\ww$ as in \eref{Ltot-EW}, we want to couple fermionic vector
and axial currents of the form $j^\mu=\ol\psi \gamma^\mu T\psi$ and
$j^{\mu5}=\ol\psi \gamma^\mu\gamma^5T'\psi$, where $\psi$ is a
Fermi multiplet with a diagonal mass matrix $M$, and  $T$ and $T'$
are hermitean $N\times N$ coupling matrices (commuting with
$\gamma^\mu$) to be determined.

The currents are generally not conserved:
\bea{cnc}
\pa_\mu j^\mu(T)= -i S([T,M]), \qquad
\pa_\mu j^{\mu5}(T')= i S^5([T',M]_+)
\eea
where $S(X)= \ol\psi X\psi$
and $S^5(X')=\ol\psi\gamma^5X'\psi$ are scalar and pseudoscalar fields,
and $[\cdot,\cdot]_+$ is the matrix anti-commutator.

For an arbitrary Lie algebra $\gg$, the most general coupling
separately satisfying the initial first-order condition \eref{LQ1} is
\bea{wtL}
\notag \wt L_1(c) &=& \sumno_a A_{a,\mu}(c)\big(j^\mu(T_a)+
j^{\mu5}_a(T'_a)\big) -i
\sumno_\ua \phi_\ua(c)\big(S([T_\ua,M])-S^5([T'_\ua,M]_+)\big)+ \\ &&+
H\big(S(V)+S^5(V')\big) \qquad\hbox{($V$ hermitean, $V'$ anti-hermitean)}.
\eea
Since for ``photons'' ($m_a=0$), $\phi_a$ do not exist, we have to
require
\bea{TA}
[T_a,M]=[T'_a,M]_+=0\qquad\hbox{if}\quad m_a=0.
\eea
In particular,
photons can couple only to mass-diagonal vector currents and to massless
axial currents, see \lref{l:diag}. For QCD with massless gluons and massive
quarks, this means that the couplings are non-chiral ($T'=0$) and the
quark masses are color-independent.

The string-independent Yukawa couplings in the second line of
\eref{wtL} are not necessary for \eref{LQ1}. Indeed $V'=0$, see \lref{l:diag},
but $V\neq0$ is needed to cancel higher-oder obstructions, see \pref{p:VKM}.
Renormalizable higher-order interactions $\wt L_n(c)$ 
($n>1$) do not exist, hence $\wt L\ww(c) = g\wt L_1(c)$. 

For the particle content of the electroweak interactions with the
specifications given in \sref{s:SM}, the cancellation of all second-order
obstructions was imposed in \cite{chir} to determine the matrices
$T_a$, $T'_a$, $V$ and $V'$ for fermion doublets ($N=2$), under the
\emph{a priori} assumption (``charge conservation'') that $T_A$, $T'_A$,
$T_Z$, $T'_Z$ are all diagonal. In \aref{a:chir}, we show that this
assumption can be dropped, because it also follows from the
cancellation conditions.

The main result of \cite{chir} is that the coupling of the $W$-bosons
is necessarily maximally chiral, that is, $T'_i=\eps T_i$ ($i=1,2$)
where $\eps$ must be a sign. (The physical sign is $\eps=-1$,
corresponding to the ``$V-A$'' structure of the weak interaction.)
A second result is that $\wt L\ww$ coincides with the local Standard
Model interaction $\wt K\ww$ (with indefinite metric) found by gauge
theory with spontaneously broken symmetry -- up to a derivative term
as in \eref{LV1} that ``carries away'' unphysical degrees of freedom
and UV-divergences that make $\wt K\ww$ non-renormalizable
\cite[Scholium~9]{chir}. 

In particular, $V'=0$ (the Higgs coupling is non-chiral), and $V=-KM$
corresponds to the proportionality between Yukawa couplings and
fermion masses. 

In \aref{a:chir}, we re-prove these results in the new setting, and
with a more algebraic perspective.
A main point of this re-working is, however, that the cancellation of
obstructions at second order as in \cite{chir} actually implies
\eref{ILQ} at all orders, see \rref{rk:Qwla}.

At the end of \aref{a:chir}, we address pertinent issues like particle
families, flavour mixing, and fermion masses and hypercharges. Essentially,
there are no new constraints from tree-level string-independence.

\section{Conclusion}
\label{s:conc}

We hope to have made a contribution to elucidate the \emph{role}
of gauge invariance -- given that gauge transformations do not affect
observables, and given the trouble conjured up by canonical
quantization of gauge fields. This includes a ``quantum answer''
to the ``classical question'' raised by Fran\c cois after an epistemic
discussion in \cite{Fra}:  
\begin{quotation} \noindent \sl
  [$\ldots$] it remains to determine what constitutes
  the proper context of justification for the electroweak theory. [$\ldots$]
  Is there a principle that would make the theory something other than a
  raw fact? 
\end{quotation}
The question naturally arose when the author had characterized the
$SU(2)$ gauge symmetry of the electroweak interactions as being
``artificial''
\footnote{\label{}See also \cite{Hi,Ki,Ly} and \rref{rk:DFM}.}
(i.e., devoid of an operational meaning because it can be ``erased'' by
a \emph{local} transformation of gauge-dependent fields into
gauge-invariant fields); while only the $U(1)$ gauge symmetry is
``substantial''. Our answer is
\begin{quotation} \noindent \sl
  Yes, there is such a  principle. And it is well
  known: it is the need of a Hilbert space in quantum theory.
\end{quotation}

Our results also include an implicit comment on \cite{Ly}, pointing
out that the ``Higgs mechanism'' has no explanatory, but an important
heuristic value ``in the context of discovery''. Lyre writes:
\begin{quotation} \noindent \sl [$\ldots$] it is almost impossible to
  invent or to discover $\mathcal{L}_{\rm GSW}'''$ from scratch. It is,
  instead, more than convenient to have some `guiding story' [$\ldots$].
\end{quotation}
Here, $\mathcal{L}_{\rm GSW}'''$ refers to the non-renormalizable 
\emph{local} Glashow-Salam-Weinberg Lagrangian with the correct
masses after SSB. One would not guess $L\ww(c)$ in \eref{Ltot-EW},
either, which yields the same S-matrix. But the recursive method of
the autonomous approach via cancellation of obstructions allows to
\emph{derive} it from scratch \cite{chir,weak}, order by order,
without a heuristic, 
but physically misleading ``guiding story'' (as if SSB were a physical
process).  

The autonomous approach to particle interactions originally set
out as a potential ``alternative to gauge theory'' \cite{Alter}
altogether. Now, after many improvements, it developped into a more
mature scheme in which gauge invariance \emph{does} plays a role.
However, it is ``downgraded'' from a ``principle'' that determines the
SM interactions, to a property that automatically comes with the 
autonomous selection of string-localized interactions $L\ww(c)$.
It signals the consistency of a model with fundamental quantum
principles, and conversely can be used to prove consistency. In the
presence of massive vector bosons, it is ``hidden''.

No massless vector potentials need to be quantized. Instead,
already quantized free fields in the Wigner representations of
physical particles are perturbed by string-localized quantum
interactions, that are constrained by the condition that the S-matrix
is string-independent. In order that this condition is fulfilled, the
(hidden) gauge invariance of $L\ww(c)$ is instrumental: it secures
invariance under the derivation $\delta_c+\omega_Q$ (the condition for
string-independence), which acts like a  (field-dependent) gauge
transformation on the string-localized fields. With techniques
exemplified for the abelian Higgs model, one can also show that 
the S-matrices are the same as in local approaches.

Specifically, in the fermionic sector of the electroweak interactions,
the autonomous interaction matches the \emph{local} but
non-renormalizable interactions of massive vector bosons with the
fermions, up to a derivative term that ``carries away'' unphysical
degrees of freedom and UV-divergences responsible for
non-renormalizability.  

But the autonomous strategy is a complete converse of the textbook
narrative, according to which one must introduce a scalar doublet with
an inverted mass term in order to break a postulated chiral gauge
symmetry with massless vector bosons. There, the ratio of coupling
constants for $\u(1)$ and  $\su(2)$ defines the Weinberg angle, which
in turn determines the ratio of the masses that the vector bosons
acquire via SSB. In order to give masses to the fermions, Yukawa
couplings to the scalar doublet have to be introduced by hand. In
contrast, in the autonomous approach, the masses can be prescribed
arbitrarily (except that $m_W$ cannot exceed $m_Z$) and determine the
Weinberg angle, the free
particles are quantized on their Wigner Hilbert spaces, and this input
determines a unique string-localized interaction \eref{Ltot-EW} (with
only one coupling constant) whose S-matrix is string-independent. A
scalar doublet never appears (except in the hidden gauge invariance
used as a tool to prove string-independence), and chirality of
fermions and their Yukawa couplings to the Higgs, proportional to
their masses, arise by necessity for string-independence.

It is true that many explicit calculations, notably in the electroweak
interaction, are much reminiscent of familiar manipulations of
classical Lagrangians. But the autonomous selection criterium is
genuinely quantum, referring to the absence of obstructions against
string-independence of the quantum S-matrix, that in turn arise 
through time-ordered quantum two-point functions. 

All assertions in this paper are understood to hold only at tree
level, and without a space-time cutoff function. This is
justified as long as we only want to determine the admissible
interactions in the autonomous approach, and to show that they
coincide with the SM interactions. For a short discussion beyond these
limitations, see \sref{s:intro}, Item~4.

The technical (and epistemic?) main result of this paper: that an underlying
(in the present setting possibly ``hidden'') gauge invariance secures the
cancellation of obstructions at all orders, is in parallel with an
analogous result obtained in the BRST setting \cite{PGI3}, where gauge
invariance secures the existence of a deformed BRST operator at all
orders \cite{PGI4}.  

\bigskip

\appendix

\section{Basic obstruction theory}
\label{a:obst}

By Wick's theorem, the difference in \eref{Om} can be expanded into
Wick products multiplied with products of $k$ propagators (contractions). The
tree-level contribution has only terms with $k=1$. Consequently,
\eref{Om} is evaluated in terms of ``two-point
obstructions'' for the linear fields involved:  
\bea{factorwise}
O_\mu(Y(y),X(x)) = \sumno_{\varphi,\chi} 
O_\mu(\varphi(y),\chi(x)) \cdot \wick{\sfrac{\pa Y}{\pa \varphi}(y)
\sfrac{\pa X}{\pa \chi}(x)}. 
\eea
Thus, all obstructions can be computed from the two-point obstructions 
$O_\mu(\varphi(y),\chi(x))$ for linear fields. We list the two-point
obstructions needed in the paper, computed with kinematic 
propagators ($c_B=c_H=c_F=0$), see \cite{AHM,weak,sQCD} in Table~A.1.
\begin{table}[h] 
$$\begin{tabular}{@{}l||c|c|c|c|c|c|@{}}
&$F^{\ka\la}(x)$ & $A^\ka(c,x)$ & $B^\ka(x)$ & $\phi(c,x)$ & $H(x)$ &$\pa^\ka H(x)$ 
\\ \hline\hline
\mystrut{11}{6} $iO_\mu(F^{\mu\nu}(y),\cdot)$ &
$\delta_\nu^{[\ka]}\pa_y^{\la]}\delta_{xy}$&
$(\delta_\nu^\ka-\pa_y^\ka  I_\nu)\delta_{xy}$ &
$\delta_\nu^\ka\delta_{xy}$ & $I_\nu\delta_{xy}$ & 0& 0 \\
\hline
\mystrut{11}{6} $iO_\mu(B^\mu(y),\cdot)$ & 0 & 0 & $m^{-2}\pa_y^\ka\delta_{xy}$
& $m^{-2}\delta_{xy}$ & 0& 0 \\
\hline
\mystrut{11}{6} $iO_\mu(\pa^\mu H(y),\cdot)$ & 0 & 0 & 0 & 0 & $-\delta_{xy}$ & $\pa_y^\ka \delta_{xy}$
  \\ \hline 
\end{tabular}
$$
\caption{Two-point obstructions $O_\mu(\varphi(y),\chi(x))$. Fields
  $B$ and $\phi$ do not exist when $m=0$. ($\delta_{xy}\equiv \delta(x-y)$)}
\label{tb:2pt} 
\end{table}
All obstructions $O(w,\cdot)$, $O(\phi(c),\cdot)$, $O(H,\cdot)$ are zero
($w$, $\phi$, and $H$ are ``inert'').  Obstructions $O(A(c),\cdot)$
appear only in the antisymmetrized form $O_{[\mu}(A(c)_{\nu]},\cdot)$,
which also vanishes ($A(c)$ is ``skew-inert'').

\begin{lemma}\label{l:deriv}
Obstruction maps $\omega_Y(\cdot)$ defined in \eref{omega} 
  acts as Wick derivations on $X(x)$:
\bea{deriv}
\omega_Y(\wick{X_1X_2}) = \wick{\omega_Y(X_1)X_2} + \wick{X_1\omega(X_2)}. 
\eea
\end{lemma}
\emph{Proof:} Immediate from \eref{factorwise} and \eref{omega}. \qed

The Master Ward Identity \cite{MWI} for ``$n$-field
obstructions'' is an identity at tree level, and can be
  imposed as a renormalization condition at loop level. After integration,
  it becomes
\footnote{\label{}Without integration,
  there is a subtlety when there are derivatives of $\delta$-functions
  in the two-point obstructions. This is detailed in \cite{MWI}, see
  also \cite[Erratum]{LV}. The subtlety is ineffective after
  integration, because an integration by parts in \eref{omega}
  precisely takes care of it.}
\bea{MWI}
\omega_Y(X_1(x_1),\dots,X_n(x_n)) &:=& i\int d^4y\,\big(T[\pa^y_\mu
Y^\mu(y)X_1(x_1)\dots X_n(x_n) - \pa^y_\mu T[Y^\mu(y)\ldots]\big) \notag \\
&=& \sumno_{k=1}^n 
T[X_1(x_1)\dots \omega_Y(X_k)(x_k)\dots X_n(x_n)].
\eea
\begin{lemma}\label{l:MWIexp} For arbitrary vector fields $Y^\mu(y)$
  of sufficiently rapid decay, it holds 
  $$ \omega_Y(e^{i\int
  dx\, X(x)})  = i \int dy\,
T [\pa^y_\mu Y^\mu(y) e^{i\int dx\, X(x)}]  = i\int dx'\, T[\omega_Y(X(x'))e^{i\int
  dx\, X(x)}].$$
\end{lemma}
\emph{Proof:} When the exponential is expanded, the first equality
is -- term by term -- the definition in \eref{MWI}, where the integral over a derivative
vanishes. The second equality is the resummation of the exponential in
the right-hand side of \eref{MWI}. \qed

\section{Electroweak fermion couplings}
\label{a:chir}

Starting with an arbitrary Lie algebra $\gg$ and fermion multiplets of
size $N$, we study the constraints on the fermionic coupling 
matrices $T_a^{(\prime)}$, $V^{(\prime)}$ in \eref{wtL} imposed
by the condition \eref{ILQ}  for $L\ww+\wt L\ww$ and $Q + \wt Q$.
As for QCD (see \sref{s:SM}, Item~3), the method using hidden gauge
invariance fails, because property (P1) fails. Therefore, we 
proceed by a direct analysis.

The initial condition \eref{LQ1} for $\wt L_1(c)$ requires the relations
\eref{TA} and determines $\wt Q_1=\sumno_aw_a(j^\mu(T_a)+ j^{\mu5}(T'_a))$. 
In view of \rref{rk:Qwla}, let $\lambda$ as in \pref{p:EW-LQ}, and  
\bea{wtQ}
\wt Q^\mu =
g\sumno_a\lambda_a\big(j^\mu(T_a)+ j^{\mu5}(T'_a)\big)
= g\sumno_a\lambda_a\sfrac{\pa \wt L_1(c)}{\pa A_{a,\mu}(c)}.
\eea
\begin{table}[h] 
$$\begin{tabular}{@{}l||c|c|c|c|@{}}
& 
$j^\mu(X)$ & $j^{\mu5}(X)$ & $S(X)$ & $S^5(X)$ 
\\ \hline\hline
\mystrut{11}{6} $\omega_{\la j(T)}$ &$-i \la j^\mu([T,X])$&$-i\la j^{5\mu}([T,X])$&$-i\la S([T,X])$&
  $-i\la S^5([T,X])$\\
  \hline
  \mystrut{11}{6} $\omega_{\la j^5(T)}$ &$-i\la j^{5\mu}([T,X])$ &$-i\la j^\mu([T,X])$&$i\la S^5([T,X]_+)$&$i\la S([T,X]_+)$\\
\end{tabular}$$
\caption{Fermionic obstructions. (See \cite[Eq.~(B.1)]{chir} and \cite[App.~A.2]{sQCD})}
\label{tb:jj} 
\end{table}
Proceeding as in QCD (Item~3 above), we only have to consider the additional
obstructions from the fermionic couplings, i.e., again string-independence requires that
\eref{ILQ-QCD} must vanish. The first three terms of \eref{ILQ-QCD}
are computed with the action of 
$\omega_Q$ on $A$, $\phi$, $H$, as given in
Table~4.3: 
\bea{f123}\notag
\delta_c(\wt L\ww) - \pa_\mu \wt Q^\mu &+& \omega_Q(\wt L\ww) =
g\Big[\sumno_a\big(\pa w_a-\pa\la_a + \rho_a+\pa I(\rho_a)\big) 
\big(j(T_a)+j^5(T'_a)\big)
\\ \notag
&&-i\sumno_\ua\big(w_\ua-\la_\ua+ I(\rho_\ua) +
\sfrac1{m_\ua^2}\sfrac{\pa Q}{\pa B_\ua}\big)\big(S([T_\ua,M])
- S^5([T'_\ua,M]_+)\big)
\\ 
&&
-\sfrac{\pa Q}{\pa(\pa H)} \big(S(V)+ S^5(V')\big)\Big].
\eea
The next term $\omega_{\wt Q}(L\ww)$ in \eref{ILQ-QCD} is zero because
$\lambda$ is inert.
For the last term, we need Table~B.1, giving
\bea{f5}\notag
&&\omega_{\wt Q}(\wt L\ww)=g^2 \sumno_a \la_a\Big[-i\sumno_{b} A_b\big(j([T_a,T_b] + [T'_a,T'_b]) + j^5([T'_a,T_b] + [T_a,T'_b])\big)
\\ \notag
&&\qquad-\sumno_{\ub}  \phi_\ub\big(S([T_a,[T_\ub,M]] +
[T'_a,[T'_\ub,M]_+]_+) - S^5([T_a,[T'_\ub,M]_+]  + [T'_a,[T_\ub,M]]_+)\big)
\\ 
&&\qquad+i H \big(S(-[T_a,V]+ [T'_a,V']_+) + S^5([T'_a,V]_+ - [T_a,V'])\big)\Big].
\eea
In \eref{f123}, we use $w_a-\la_a+I(\rho_a)=0$ by \eref{wGT} and
$\pa w_a-\pa \la_a + \rho_a+\pa I(\rho_a) = gf_{abc}\la_b A_c$ 
by \eref{pawGT}. Then, and with $\frac{\pa Q}{\pa B}$ and
$\frac{\pa Q}{\pa(\pa H)}$ as given in Table~4.3, \eref{ILQ-QCD}
is a sum of terms of the form $\lambda_a$ times 
$A_bj,A_bj^5,\phi_\ub S, \phi_\ub S^5, HS,$ or $H S^5$, which must vanish separately.

This yields six conditions for string-independence on top of \eref{TA}, 
that can be written as matrix relations:
\bea{jj}
i\big([T_a,T_b]+ [T'_a,T'_b]\big) &=& \sumno_c f_{abc} \cdot T_c,
\\ \label{jj5}
i\big([T_a,T'_b]+[T'_a,T_b]\big)
&=& \sumno_c f_{abc}\cdot T'_c,
\\ \label{phiS} \sumno_\uc \sfrac{im_\ub}{m_\uc} \ga_{\ub a\uc}
[T_\uc,M] -[T_a,[T_\ub,M]]-[T'_a,[T'_\ub,M]_+]_+&=&
Km_a^2\delta_{a\ub}V,
\\ \label{phiS5} \sumno_\uc \sfrac{im_\ub}{m_\uc} \ga_{\ub a\uc}
[T'_\uc,M]_+  - [T_a,[T'_\ub,M]_+]-[T'_a,[T_\ub,M]]_+&=&
-Km_a^2\delta_{a\ub}V',
\\ \label{HS} - [T_a,V]+[T_a',V']_+ &=& K[T_a,M] ,   
\\ \label{HS5} - [T'_a,V]_++[T_a,V'] &=& K[T'_a,M]_+.\qquad
\eea
\eref{jj} $\pm$ \eref{jj5} assert that both $T_a^\pm:= T_a\pm T'_a$
are representations of the Lie algebra $\gg$: 
\bea{pipma}i[T^\pm_a,T^\pm_b] = \sumno_c f_{abc}
\cdot T^\pm_c \qquad\Leftrightarrow\qquad  iT_a^\pm = \pi^\pm(\tau_a).
\eea
\eref{HS} $\pm$ \eref{HS5} can be written as intertwining relations
between $\pi^+$ and $\pi^-$:
\bea{intertwine} (KM+V\pm V')\pi^\pm(\tau)=\pi^\mp(\tau)(KM+V\pm V').
\eea

We now specify $\gg=\u(1)\oplus\su(2)$ for the electroweak
interactions with structure constants
$f_{A12}=\sin\theta$, $f_{12Z}=\cos\theta$ both $\neq0$.
We consider Fermi doublets, $N=2$, in a mass eigenbasis with masses
$\mu_2>\mu_1$.  (The case of equal masses $\mu_1=\mu_2$ requires a
separate analysis.) 

Let $T^\pm_0:=\cos\theta \,T_A^\pm\,-\sin\theta\, T_Z^\pm=-i\pi^\pm(\tau_0)$
and $T^\pm_3:=\sin\theta \,T_A^\pm+\cos\theta\, T_Z^\pm$. Then
\bea{pipmi}
T_0^\pm = -i\pi^\pm (\tau_0),\qquad T_i^\pm=-\sfrac12\pi^\pm(\sigma_i) 
\qquad (i=1,2,3),
\eea
where $\tau_0$ is the (imaginary) generator of $\u(1)$. $\pi^+$ and $\pi^-$
are (each) either trivial on $\su(2)$ or unitarily equivalent to the
identical representation. In the latter case, the ``hypercharge matrix''
$T^+_0$ resp.\ $T^-_0$ is a multiple of $\eins$.

The next Lemma asserts that ``charge conservation'' in the
interaction vertices, that was assumed \emph{a priori} in \cite{chir},
is actually a consequence of string-independence in the
autonomous approach that does not assume 
any kind of \emph{a priori} symmetry, including charge conservation. 
\begin{lemma}\label{l:diag} (``Charge conservation'').
  If $\mu_2>\mu_1$, the matrices $T_A$, $T_A'$, $T_Z$, and $T'_Z$ are
  all diagonal. If the neutrino is massless ($\mu_1=0$), i.e.,
  $M =\bpm 0&\\[-2mm]&\mu_2\epm=:\mu_2P$, then $T'_A$ is a multiple of
$P^\perp:=\eins-P$, otherwise $T'_A=0$. The matrix $V$ is diagonal, and $V'=0$.
\end{lemma}
$V'=0$ means that the Higgs couples
non-chirally to the fermions. 

\noindent \emph{Proof of \lref{l:diag}:} $[T_A,M]=0$ in \eref{TA} implies that $T_A$ is diagonal.
If $\mu_1>0$, then $[T'_A,M]_+=0$  implies $T'_A=0$. If $\mu_1=0$, it
implies that $T'_A\propto P^\perp$. 

If $\pi^+$ and $\pi^-$ are both nontrivial on $\su(2)$, then both
$T^\pm_0\propto \eins$. If, say, $\pi^+$ is nontrivial and 
$\pi^-$ is trivial on $\su(2)$, then  $T^+_0\propto\eins$
and $T^-_3=0$. If $\pi^+$ and $\pi^-$ are both trivial on $\su(2)$,
then both $T^\pm_3=0$. In each case, we have four linearly independent
combinations of $T_A$, $T_A'$, $T_Z$, and $T'_Z$ that are diagonal.

The last statement follows from \eref{phiS} and \eref{phiS5} with
$a=b=Z$. \qed

With the assumption of ``charge conservation'' fulfilled by
\lref{l:diag}, there is no need to repeat the detailed determination
of the coupling matrices $T_a$ and $T'_a$ done in \cite{chir}.
Nevertheless, we reproduce here the central results of
\cite{chir}, but in a  more ``algebraic'' way.

When we discard the physically uninteresting case that the W-bosons do not
couple to fermions, \pref{p:chir} is the autonomous prediction of
maximal chirality of the weak interactions \cite{chir}: it is another necessary
condition for string-independence.
\begin{prop}\label{p:chir} (``The chirality theorem'').
If $\mu_2>\mu_1$, then either $\pi^+$ or $\pi^-$ is trivial on $\su(2)$.
If both are trivial, then the W-bosons cannot couple to the fermions.
\end{prop}
The condition $\mu_1\neq\mu_2$ is crucial. If $\mu_1=\mu_2$,
\lref{l:diag} does not hold. Then, $T'_a=0$ for all $a$
and $V=V'=0$ is a solution, and $\pi^+=\pi^-$ may be nontrivial on $\su(2)$.

\emph{Proof of \pref{p:chir}:} Assume that $\pi^\pm$ are \emph{both}
nontrivial on $\su(2)$. Then both $T^\pm_0\propto\eins$,
hence $T_0=\cos\theta \,T_A-\sin\theta\, T_Z$ and
$T'_0=\cos\theta \,T'_A-\sin\theta\, T'_Z$ are multiples of $\eins$.

$T^\pm_3 $ are both diagonal, hence  
$\propto\sigma_3$ by \eref{pipmi}, but may differ by a sign: $T_3^+=\pm T_3^-$. 

We first exclude the option $T_3^+=T_3^-$. Then $T'_3=0$, hence along with
$T'_A$, also $T'_Z$ and $T'_0$ would be multiples of $P^\perp$. But
$T'_0$ is a multiple of $\eins$, hence $T'_0=0$, hence also $T'_A=T'_Z=0$.

Because $T_Z$ is diagonal and $T'_Z=0$, \eref{phiS}
and \eref{phiS5} with $a=b=Z$ would imply that $V=V'=0$. But this conflicts with
\eref{intertwine} because $M$ is not a multiple of a unitary.

The option $T_3^+=-T_3^-$, hence $T_3=0$ and $T'_3=\pm \frac12\sigma_3$,
is easily excluded if $\mu_1>0$: Namely, if $T'_A=0$ and
$T'_0$ is a multiple of $\eins$, then also $T'_3$ must be a multiple of $\eins$.

If $\mu_1=0$, the conditions $T'_0\propto \eins$ and
$T'_A\propto P^\perp$ determine 
$T'_Z=\pm\frac1{2\cos\theta}\bpm \cos2\theta &\\[-2mm]&1\epm$.
On the other hand, the intertwiners in \eref{intertwine} are diagonal by
\lref{l:diag}, and cannot switch the sign of $\sigma_3$, unless
$KM+V=V'=0$. But \eref{phiS} with $a=b=Z$ would give $-4T_Z'^2M = Km_Z^2 V$.
Contradiction.

Finally, if both representations $\pi^\pm$
are trivial on $\su(2)$, then $T_i=T'_i=0$ ($i=1,2,3$).
\qed

String-independence does not decide whether $\pi^+$ or $\pi^-$ is
nontrivial on $\su(2)$. Nature chooses $\pi^-$ (the ``$V-A$''
structure of the weak interaction). Since $T^-_3$ is diagonal by
\lref{l:diag}, the mass eigenbasis of the vector bosons can be fixed such that
$\pi^-$ is actually the identical representation of $\su(2)$.

We now show that string-independence requires the Yukawa
couplings $HS(V)$ in \eref{wtL} to be proportional to the Fermi masses. 
\begin{prop}\label{p:VKM} (``Yukawa couplings'').
  In the chiral case, it holds $V=-KM$.
\end{prop}
\emph{Proof:} Recall $V'=0$. By \eref{intertwine}, $KM+V$ intertwines 
$\pi^+$ and $\pi^-$ which are inequivalent as representations of
$\su(2)$. Hence it must vanish. \qed

\bigskip

Just for the sake of completeness, we present the final result of \cite{chir}:
$T'_A=0$ because the neutrino was assumed massive, cf.~\eref{TA}, and
$$T_i=-T'_i = -\sfrac14 \sigma_i \quad (i=1,2), \quad T_A = \bpm 0& \\
&s \epm, \quad T_Z=  \sfrac1{4c}\bpm -1 & \\ & 1-4s^2 \epm,
\quad T'_Z= \sfrac 1{4c}\sigma_3, $$
where $s\equiv \sin\theta$, $c\equiv \cos\theta$. Conditions
\eref{jj}--\eref{phiS5} can be verified one by one by elementary
matrix calculus. The coupling matrix $T_A$ asserts that the electric
unit of charge is $e=g\,\sin\theta$. The eigenvalues of
$T^\pm_0=\cos\theta\, T_A-\sin\theta\, T^\pm_Z$ are the hypercharges
times $-\frac{g'}g=-\tan\theta$. The displayed solution is unique,
except that one may add a multiple of $\eins$ to $T_A$. Thus, also
quark doublets are covered. Their hypercharges are not fixed, only
the difference for the upper and lower components of doublets is
necessarily $1$ by string-independence. The issue of the need for anomaly
  cancellation within a family, that would constrain hypercharges
  individually, cannot be addressed at tree level. Yet, as the
  relevant fermion loops do not contain string-localized propagators, we expect the
  same anomalies as in the local approach.

Strong and weak couplings of quarks (\sref{s:SM}, Items~3 and 6)
  in combination do not produce
mutual obstructions because color and flavour matrices commute. Also
several families of lepton and quark doublets 
produce no new obstructions, and no constraints on the masses
  and on the CKM mixing matrix: one can see that the matrices (in an
  eigenbasis of Fermi doublets of arbitrary masses)
  $$T_A=\bpm 0&\\[-2mm] &s\epm\otimes \eins_n, \qquad \pi^+(\sigma_i)=0,
  \quad\pi^-(\sigma_i) = U^* (\sigma_i\otimes\eins_n)U \qquad (i=1,2,3),$$
where $U = \bpm U_d & \\[-2mm]&U_u\epm$ such that $V_{\rm CKM}:= U_d^*U_u$
is the CKM matrix, solves the constraints \eref{jj}--\eref{HS5}
for arbitrary $U_d$ and $U_c$. 

\bigskip

{\bf Acknowledgment.} We thank J. Mund (who suggested the
``converse'' part of \pref{p:ILQ}), J.M. Gracia-Bond\'ia and 
B. Schroer for valuable comments on earlier versions of the paper. We
are grateful to the peer reviewers for very valuable critique and suggestions.


\begin{thebibliography}{99} \itemsep0mm
\bibitem
  {BFetal} P. Berghofer, J. Fran\c cois, Simon Friederich, H. Gomes,
  G. Hetzroni, A. Maas, R. Sondenheimer: Gauge symmetry, symmetry
  breaking, and gauge-invariant approaches.
  In: Elements in the Foundations of Contemporary Physics R. Dawid and
  J. Wells, eds., Cambridge Univ. Press (2023), 
  doi:10.1017/9781009197236.
\bibitem
  {MWI} F.-M. Boas, M. Dütsch: The Master Ward Identity.
  Rev. Math. Phys. 14 (2002) 977--1049,
  doi:10.1142/S0129055X02001454.
 \bibitem
   {Bu} D. Buchholz: Gauss’ law and the infraparticle problem.
   Phys. Lett. B 174 (1986) 331--334,
   doi:10.1016/0370-2693(86)91110-X.
\bibitem
  {Di} P.A.M. Dirac: Gauge-invariant formulation of Quantum Electrodynamics.
  Can. J. Phys. 33 (1955) 650--660,
  doi:10.1139/p55-081.
\bibitem
  {PGI4} M. Dütsch: Proof of perturbative gauge invariance for tree
  diagrams to all orders.
  Ann. Phys. (Leipzig) 14 (2005) 438--61,
  doi:10.1002/andp.20055170703.
\bibitem
  {Ear1} J. Earman: Curie’s principle and spontaneous symmetry breaking.
  Intl. Studies Philos. Sci. 18 (2004) 173--198,
  doi:10.1080/0269859042000311299.
\bibitem
  {Ear2} J. Earman: Explaining the Aharonov-Bohm effect. 
  https://philsci-archive.pitt.edu/ id/eprint/22970.
\bibitem
   {EG} H. Epstein, V. Glaser: The role of locality in perturbation
   theory.
   Ann. Inst. H. Poinc., Phys. Theor. 19A (1973) 211--295.
\bibitem
  {Fra} J. Fran\c cois: Artificial vs substantial gauge symmetries: a
  criterion and an application to the electroweak model.
  Phil. of Sci. 86 (2019) 472--496, arXiv:1801.00678,
  doi:10.1086/703571.
\bibitem
  {FMS1} J. Fröhlich, G. Morchio, F. Strocchi: Charged sectors and
  scattering states in quantum electrodynamics.
  Ann. Phys. 119 (1979) 241--284,
  doi:10.1016/0003-4916(79)90187-8.
\bibitem
  {FMS2} J. Fröhlich, G. Morchio, F. Strocchi: Higgs phenomenon without
  symmetry breaking parameter.
  Nucl. Phys. B190 [FS3] (1981) 553--582,
  doi: 10.1016/0370-2693(80)90594-8.
  \bibitem
    {G} Ch. Gass: Renormalization in string-localized field theories:
    a microlocal analysis.
    Ann. Henri Poincar\'e 23 (2022) 3493--3523,
    doi:10.1007/s00023-022-01171-3.
  \bibitem
  {GGM} Ch. Gass, J.M. Gracia-Bond\'ia,  J. Mund:
  Revisiting the Okubo–Marshak argument.
  Symmetry 13 (2021) 1645,
  doi:10.3390/sym13091645.
\bibitem
  {chir} J.M. Gracia-Bondía, J. Mund, J.C. Várilly: The chirality theorem.
  Ann. Henri Poincaré 19 (2018) 843--874,
  doi:10.1007/s00023-017-0637-3.
\bibitem
  {HILLM} T. Heinzl, A. Ilderton, K. Langfeld, M. Lavelle, D. McMullan:
  Ice limit of Coulomb gauge Yang-Mills theory.
  Phys. Rev. D78 (2008) 074511.
\bibitem
  {weak} J.M. Gracia-Bondía, K.-H. Rehren, J.C. Várilly: The full
  electroweak interaction: an autonomous account.
  Ann. Henri Poincaré 26 (2025) 4529--4574,
  doi:10.1007/s00023-025-01540-8.
\bibitem
  {sQCD} I. Hemprich, K.-H. Rehren: Dressed fields for quantum chromodynamics.
  Lett.\ Math.\ Phys.\ 115 (2025) 112,
  doi:10.1007/s11005-025-02002-2.
\bibitem
  {Hi} P.W. Higgs: Spontaneous symmetry breakdown without massless bosons.
  Phys. Rev. 145 (1966) 1156--1163,
  doi:10.1103/PhysRev.145.1156.
\bibitem
  {Ki} T.W.B. Kibble: Symmetry breaking in non-abelian gauge theories.
  Phys. Rev. 155 (1967) 1554-1561,
  doi:10.1103/PhysRev.155.1554.
\bibitem
  {Jo} P. Jordan: Der gegenwärtige Stand der Quantenelektrodynamik (in German).
  In: Talks and discussions of the theoretical-physical
  conference in Kharkov (19–25 May 1929). Phys. Z. XXX (1929) 700--712.
\bibitem
  {LM} M. Lavelle, D. McMullan: Constituent quarks from QCD.
  Phys. Rept. 279 (1997) 1--65.
\bibitem
  {Ly} H. Lyre: Does the Higgs Mechanism exist?
  Intl. Studies Philos. Sci. 22 (2008) 119--133,
  doi:10.1080/02698590802496664.
\bibitem
  {infra} J. Mund, K.-H. Rehren, B. Schroer:   Infraparticle quantum
  fields and the formation of photon clouds. 
  JHEP 04 (2022) 083,
  doi:10.1007/JHEP04(2022)083.
\bibitem
  {AHM} J. Mund, K.-H. Rehren, B. Schroer: How the Higgs potential got its shape.
  Nucl. Phys. B 987 (2023) 116109,
  doi:10.1016/j.nuclphysb.2023.116109.
\bibitem
  {MR} J. Mund, K.-H. Rehren: Construction of string-localized massive QED.
  In progress.
\bibitem
  {OP} V.I. Ogievetski, I.V. Polubarinov: On the meaning of gauge invariance.
  Nuovo Cim. 23 (1962) 173--180,
  doi:10.1007/BF02733552.
\bibitem
  {LV} K.-H. Rehren:
  On the effect of derivative interactions in quantum field theory. 
  Lett.\ Math.\ Phys.\ 115 (2025) 3, and Erratum: Lett. Math. Phys. 115 (2025) 72,
  doi:10.1007/s11005-024-01889-7.
\bibitem
  {aut} K.-H. Rehren, L.T. Cardoso, C. Gass, J.M. Gracia-Bond\'ia,
  B. Schroer, J.C. V\'arilly:
  sQFT: an autonomous explanation of the interactions of quantum particles.
  Found.\ Phys.\ 54 (2024) 57, 
  doi:10.1007/s10701-024-00795-1.
\bibitem
  {why} S. Rivat: Wait, why gauge?
  Br. J. Philos. Sci. (2024),
  doi:10.1086/727736.
\bibitem
  {Sa} P.L. Saldanha: Local description of the Aharonov-Bohm effect with a quantum
  electromagnetic field.
  Found. Phys. 51 (2021) 6,
  doi:10.1007/s10701-021-00414-3.
\bibitem
  {PGI3}  G. Scharf: Gauge Field Theories: Spin One and Spin Two
  (Dover, 2016).
\bibitem
  {Alter} B. Schroer: An alternative to the gauge theoretic setting.
  Found. Phys. 41 (2011) 1543--1568,
  doi:10.1007/s10701-011-9567-y.
\bibitem
  {Schroer}B. Schroer: Peculiarities of massive vector mesons and
  their zero mass limits. 
  Eur. Phys. J. C75 (2015) 365,
  doi:10.1140/epjc/s10052-015-3497-3.
\bibitem
  {Wein} S. Weinberg: Feynman rules for any spin.
  Phys. Rev. 133 (1964) B1318--B1332,
  doi:10.1103/PhysRev.133.B1318.
  and Part II: Massless particles.
  Phys. Rev. 134 (1964) B882--B896,
  doi:10.1103/PhysRev.134.B882.
\end{thebibliography}
\end{document}